\documentclass[11pt]{article}
\usepackage[margin=1in]{geometry}

\usepackage[english]{babel}
\usepackage[numbers,sort]{natbib}
\usepackage[colorinlistoftodos]{todonotes}
\renewcommand{\epsilon}{\varepsilon}

\usepackage[colorlinks=true,citecolor=blue,urlcolor=blue]{hyperref}

\usepackage{amsmath}

\usepackage{amsthm}
\usepackage{amssymb}
\usepackage{amsfonts}
\usepackage{dsfont}
\usepackage{thm-restate}
\usepackage{nicefrac}
\usepackage{footnote}
\usepackage[linesnumbered,boxed, vlined]{algorithm2e}
\usepackage{caption}
\usepackage{subcaption}
\usepackage{tikzsymbols}
\usepackage{tikz}

\tikzstyle{vertex}=[circle, draw,fill=gray!30, inner sep=0pt, minimum size=16pt]
\tikzstyle{svertex}=[circle, draw,fill=gray!30, inner sep=0pt, minimum size=10pt]
\tikzstyle{sgvertex}=[circle, draw,fill=gray!15, inner sep=0pt, minimum size=10pt]
\tikzstyle{ssgvertex}=[circle, draw,fill=gray!15, inner sep=0pt, minimum size=6pt]
\tikzstyle{sssgvertex}=[circle, draw,fill=gray!15, inner sep=0pt, minimum size=4pt]
\tikzstyle{sssagvertex}=[circle, draw,fill=gray!15, inner sep=0pt, minimum size=3.2pt]
\tikzstyle{ssssgvertex}=[circle, draw,fill=gray!15, inner sep=0pt, minimum size=2pt]

\usepackage[capitalize]{cleveref}
\usepackage{algorithm2e}
\usepackage{algorithmic}
\usepackage{enumerate}

\newtheorem{claim}{Claim}
\newcommand{\wt}[1]{\widetilde{#1}}

\DeclareMathOperator{\TV}{TV}

\newcommand{\e}{\epsilon}
\newcommand{\expect}{{\bf E}}

\theoremstyle{plain}
\newtheorem{thm}{Theorem}[section]

\newtheorem{lem}[thm]{Lemma}

\newtheorem{claim}[thm]{Claim}

\newtheorem*{cla*}{Claim}

\newtheorem{Def}[thm]{Definition}
\newtheorem{obs}[thm]{Observation}

\renewcommand{\geq}{\geqslant}
\renewcommand{\leq}{\leqslant}

\newcommand{\eps}{\varepsilon}

\begin{document}

\title{Online Matching with General Arrivals}
 \author{
   Buddhima Gamlath \\
   EPFL\\
  \texttt{buddhima.gamlath@epfl.ch}
     \and
   Michael Kapralov\footnote{Supported in part by ERC Starting Grant 759471.} \\
  EPFL\\
   \texttt{michael.kapralov@epfl.ch}
     \and
   Andreas Maggiori\\
  EPFL\\
   \texttt{andreas.maggiori@epfl.ch}
     \and
   Ola Svensson\\
  EPFL\\
   \texttt{ola.svensson@epfl.ch}
   \and 
   David Wajc\footnote{Work done in part while the author was visiting EPFL. Supported in part by NSF grants CCF-1618280, CCF-1814603, CCF-1527110, NSF CAREER award CCF-1750808, and a Sloan Research Fellowship.} \\
   CMU \\
   \texttt{dwajc@cs.cmu.edu}
 }
\date{}

\maketitle

\pagenumbering{gobble}
\begin{abstract}
  The online matching problem was introduced by Karp, Vazirani and Vazirani nearly three decades ago. In that seminal work, they studied this problem in bipartite graphs with vertices arriving only on one side, and presented optimal deterministic and randomized algorithms for this setting.
  In comparison, more general arrival models, such as edge arrivals and general vertex arrivals,  have proven more challenging and positive results are known only for various relaxations of the problem. In particular, even the basic question of whether randomization allows one to beat the trivially-optimal deterministic competitive ratio of $\nicefrac{1}{2}$ for either of these models was open. In this paper, we resolve this question for both these natural arrival models, and show the following. 
  \begin{enumerate}
      \item For edge arrivals, randomization does not help --- no randomized algorithm is better than $\nicefrac{1}{2}$ competitive.
      \item For general vertex arrivals, randomization helps --- there exists a randomized $(\nicefrac{1}{2}+\Omega(1))$-competitive online matching algorithm.
  \end{enumerate}
\end{abstract}

\newpage
\pagenumbering{arabic}

%!TEX root = ./main.tex
\section{Introduction}\label{sec:intro}
Matching theory has played a prominent role in the area of combinatorial optimization, with many applications \cite{lovasz2009matching,schrijver2003combinatorial}. Moreover, many fundamental techniques and concepts in combinatorial optimization can trace their origins to its study, including the primal-dual framework \cite{kuhn1955hungarian}, proofs of polytopes' integrality beyond total unimodularity \cite{edmonds1965maximum}, and even the equation of efficiency with polytime computability \cite{edmonds1965paths}.

Given the prominence of matching theory in combinatorial optimization, it comes as little surprise that the maximum matching problem was one of the first problems studied from the point of view of online algorithms and competitive analysis. 
In 1990, \citet*{karp1990optimal} introduced the online matching problem, and studied it under one-sided bipartite arrivals. For such arrivals, 
Karp et al.~noted that the trivial $\nicefrac{1}{2}$-competitive greedy algorithm (which matches any arriving vertex to an arbitrary unmatched neighbor, if one exists) is optimal among deterministic algorithms for this problem. 
More interestingly, they provided an elegant randomized online algorithm for this problem, called \textsc{ranking}, which achieves an optimal $(1-\nicefrac{1}{e})$ competitive ratio. (This bound has been re-proven many times over the years \cite{birnbaum2008line,goel2008online,devanur2013randomized,eden2018economic,feige2018tighter}.)
Online matching and many extensions of this problem under one-sided bipartite vertex arrivals were widely studied over the years, both under adversarial and stochastic arrival models.
See recent work \cite{huang2018match,huang2018online,huang2019tight,cohen2018randomized} and the excellent survey of \citet{mehta2013online} for further references on this rich literature.

Despite our increasingly better understanding of one-sided online bipartite matching and its extensions, the problem of online matching under more general arrival models, including edge arrivals and general vertex arrivals, has remained staunchly defiant, resisting attacks. In particular, the basic questions of whether the trivial $\nicefrac{1}{2}$ competitive ratio is optimal for the adversarial edge-arrival and general vertex-arrival models have remained tantalizing open questions in the online algorithms literature. In this paper, we answer both of these questions.

\subsection{Prior Work and Our Results}
Here we outline the most relevant prior work, as well as our contributions. Throughout, we say an algorithm (either randomized or fractional) has \emph{competitive ratio} $\alpha$, or equivalently is \emph{$\alpha$-competitive}, if the ratio of the algorithm's value (e.g., expected matching size, or overall value, $\sum_e x_e$) to OPT is at least $\alpha\leq 1$ for all inputs and arrival orders. As is standard in the online algorithms literature on maximization problems, we use upper bounds (on $\alpha$) to refer to hardness results, and lower bounds to positive results.

\vspace{-0.1cm}
\paragraph{Edge Arrivals.} Arguably the most natural, and the least restricted, arrival model for online matching is the edge arrival model. In this model, edges are revealed one by one, and an online matching algorithm must decide immediately and irrevocably whether to match the edge on arrival, or whether to leave both endpoints free to be possibly matched later.

On the hardness front, the problem is known to be strictly harder than the one-sided vertex arrival model of \citet{karp1990optimal}, which admits a competitive ratio of $1-\nicefrac1e \approx 0.632$.
In particular, \citet{epstein2013improved} gave an upper bound of $\frac{1}{1+\ln2} \approx 0.591$ for this problem,
recently improved by \citet{huang2019tight} to $2-\sqrt{2}\approx 0.585$. (Both bounds apply even to online algorithms with preemption; i.e., allowing edges to be removed from the matching in favor of a newly-arrived edge.)
On the positive side, as pointed out by \citet{buchbinder2018online}, the edge arrival model has proven challenging, and results beating the $\nicefrac{1}{2}$ competitive ratio were only achieved under various relaxations, including: random order edge arrival 
\cite{guruganesh2017online}, bounded number of arrival batches \cite{lee2017maximum}, on trees, either with or without preemption \cite{tirodkar2017maximum,buchbinder2018online}, and for bounded-degree graphs \cite{buchbinder2018online}.
The above papers all asked whether there exists a randomized $(\nicefrac{1}{2}+\Omega(1))$-competitive algorithm for adversarial edge arrivals (see also Open Question 17 in Mehta's survey \cite{mehta2013online}).

In this work, we answer this open question, providing it with a strong negative answer. In particular, we show that no online algorithm for \emph{fractional} matching (i.e., an algorithm which immediately and irrevocably assigns values $x_e$ to edge $e$ upon arrival such that $\vec{x}$ is in the fractional matching polytope $\mathcal{P}=\{\vec{x}\geq \vec{0}\mid \sum_{e\ni v}x_e \leq 1\,\, \forall v\in V\}$) is better than $\nicefrac{1}{2}$ competitive. As any randomized algorithm induces a fractional algorithm with the same competitive ratio, this rules out any randomized online matching algorithm which is better than deterministic algorithms.

\begin{thm}\label{thm:edge-intro}
	No fractional online algorithm is $\nicefrac{1}{2}+\Omega(1)$ competitive for online matching under adversarial edge arrivals, even in bipartite graphs.
\end{thm}

This result shows that the study of relaxed variants of online matching under edge arrivals is not only justified by the difficulty of beating the trivial bound for this problem, 
but rather by its \emph{impossibility}.

\vspace{-0.1cm}
\paragraph{General Vertex Arrivals.} In the online matching problem under vertex arrivals, 
vertices are revealed one at a time, together with their edges to their previously-revealed neighbors. An online matching algorithm must decide immediately and irrevocably upon arrival of a vertex whether to match it (or keep it free for later), and if so, who to match it to.
The one-sided bipartite problem studied by \citet{karp1990optimal} is precisely this problem when all vertices of one side of a bipartite graph arrive first. As discussed above, for this one-sided arrival model, the problem is thoroughly understood (even down to lower-order error terms \cite{feige2018tighter}). 
\citet{wang2015two} proved that general vertex arrivals are strictly harder than one-sided bipartite arrivals, providing an upper bound of $0.625 < 1-\nicefrac{1}{e}$ for the more general problem, later improved by \citet{buchbinder2018online} to $\frac{2}{3+\phi^2}\approx 0.593$. 
Clearly, the general vertex arrival model is no harder than the online edge arrival model
but is it \emph{easier}? The answer is ``yes'' for \emph{fractional} algorithms, as shown by combining our \Cref{thm:edge-intro} with the $0.526$-competitive fractional online matching algorithm under general vertex arrivals of \citet{wang2015two}. For \emph{integral} online matching, however, the problem has proven challenging, and the only positive results for this problem, too, are for various relaxations, such as 
restriction to trees, either with or without preemption \cite{tirodkar2017maximum,chiplunkar2015randomized,buchbinder2018online}, for bounded-degree graphs \cite{buchbinder2018online}, or (recently) allowing vertices to be matched during some known time interval \cite{huang2018match,huang2019tight}.
	
We elaborate on the last relaxation above. In the model recently studied by \citet{huang2018match,huang2019tight} vertices have both arrival and departure times, and edges can be matched whenever both their endpoints are present. (One-sided vertex arrivals is a special case of this model with all online vertices departing immediately after arrival and offline vertices departing at $\infty$.) 
We note that any $\alpha$-competitive online matching under general vertex arrivals is $\alpha$-competitive in the less restrictive model of Huang et al.
As observed by Huang et al., for their model an optimal approach might as well be greedy; i.e., an unmatched vertex $v$ should always be matched at its departure time if possible. In particular, \citet{huang2018match,huang2019tight}, showed that the \textsc{ranking}  algorithm of Karp et al.~is optimal in this model, giving a competitive ratio of $\approx 0.567$.
For general vertex arrivals, however, \textsc{ranking} (and indeed any maximal matching algorithm) is no better than $\nicefrac{1}{2}$ competitive, as is readily shown by a path on three edges with the internal vertices arriving first. Consequently, new ideas and algorithms are needed.

The  natural open question for general vertex arrivals is whether a competitive ratio of $(\nicefrac{1}{2}+\Omega(1))$ is achievable by an \emph{integral} randomized algorithm, without any assumptions (see e.g., \cite{wang2015two}).
In this work, we answer this question in the affirmative: 

\begin{thm}\label{thm:vertex-intro}
	There exists a $\left(\nicefrac{1}{2}+\Omega(1)\right)$-competitive randomized online matching algorithm for 
	general adversarial vertex arrivals.
\end{thm}

\subsection{Our Techniques}
\label{sec:techniques}

\paragraph{Edge Arrivals.}
All prior upper bounds in the online literature \cite{karp1990optimal,epstein2013improved,huang2019tight,buchbinder2018online,feige2018tighter} can be rephrased as upper bounds for \emph{fractional} algorithms; i.e., algorithms which immediately and irrevocably assign each edge $e$ a value $x_e$ on arrival, so that $\vec{x}$ is contained in the fractional matching polytope, $\mathcal{P}=\{\vec{x}\geq \vec{0} \mid \sum_{e\ni v} x_e \leq 1\,\,\forall v\in V\}$. 
With the exception of \cite{buchbinder2018online}, the core difficulty of these hard instances is uncertainty about ``identity'' of vertices (in particular, which vertices will neighbor which vertices in the following arrivals). 
Our hardness instances rely on uncertainty about the ``time horizon''. In particular, the underlying graph, vertex identifiers, and even arrival order are known to the algorithm, but the number of edges of the graph to be revealed (to arrive) is uncertain. Consequently, an $\alpha$-competitive algorithm must accrue high enough value up to each arrival time to guarantee a high competitive ratio at all points in time. As we shall show, for competitive ratio  $\nicefrac{1}{2}+\Omega(1)$, this goal is at odds with the fractional matching constraints, and so such a competitive ratio is impossible.
In particular, we provide a family of hard instances and formulate their prefix-competitiveness and matching constraints as linear constraints to obtain a linear program whose objective value bounds the optimal competitive ratio. Solving the obtained LP's dual, we obtain by weak duality the claimed upper bound on the optimal competitive ratio.
%See, e.g., \cite{azar2017online} for more examples of this approach.

\paragraph{General Vertex Arrivals.}
Our high-level approach here will be to round online a fractional online matching algorithm's output, specifically that of \citet{wang2015two}. While this approach sounds simple, there are several obstacles to overcome. First, the fractional matching polytope is not integral in general graphs, where a fractional matching may have value, $\sum_e x_e$, some $\nicefrac{3}{2}$ times larger than the optimal matching size. (For example, in a triangle graph with value $x_e=\nicefrac{1}{2}$ for each edge $e$.)
Therefore, any general rounding scheme must lose a factor of $\nicefrac{3}{2}$ on the competitive ratio compared to the fractional algorithm's value, and so to beat a competitive ratio of  $\nicefrac{1}{2}$ would require an online fractional matching with competitive ratio $> \nicefrac{3}{4} > 1-\nicefrac{1}{e}$, which is impossible. To make matters worse, even in bipartite graphs, for which the fractional matching polytope is integral and offline lossless rounding is possible  \cite{ageev2004pipage,gandhi2006dependent}, \emph{online} lossless rounding of fractional matchings is impossible, even under one-sided vertex arrivals  \cite{cohen2018randomized}. 

Despite these challenges, we show that a slightly better than $\nicefrac{1}{2}$-competitive fractional matching computed by the algorithm of \cite{wang2015two} can be rounded online without incurring too high a loss, yielding $(\nicefrac{1}{2}+\Omega(1))$-competitive randomized algorithm for online matching under general vertex arrivals.

To outline our approach, we first consider a simple method to round matchings online. When vertex $v$ arrives, we pick an edge $\{u,v\}$ with probability $z_u = x_{uv}/\Pr[u \mbox{ free when $v$ arrives}]$, and add it to our matching if $u$ is free.

If $\sum_u z_u \leq 1$, this allows us to pick at most one edge per vertex and have each edge $e=\{u,v\}$ be in our matching with the right marginal probability, $x_e$, resulting in a lossless rounding. Unfortunately, we know of no better-than-$\nicefrac12$-competitive fractional algorithm for which this rounding guarantees $\sum_u z_u \leq 1$.

However, we observe that, for the correct set of parameters, the fractional matching algorithm of
Wang and Wong~\cite{wang2015two} makes $\sum_u z_u$ close to one, while still ensuring a better-than-$\nicefrac{1}{2}$-competitive fractional solution. 
Namely, as we elaborate later in \Cref{sec:imrpoved-algorithm}, we set the parameters of their algorithm so that $\sum_u z_u \leq 1 + O(\e)$, while retaining a competitive ratio of $1/2 + O(\e)$.
Now consider the same rounding algorithm with normalized probabilities: I.e., on $v$'s 
arrival, sample a neighbor $u$ with probability $z'_u = z_u/\max\{1, \sum_u z_u\}$ and match if $u$ is free. 
As the sum of $z_u$'s is slightly above one in the worst case, this approach does not drastically reduce the 
competitive ratio. 
But the normalization factor is still too significant compared to the competitive ratio of the fractional solution, 
driving the competitive ratio of the rounding algorithm slightly below $1/2$.

To account for this minor yet significant loss, we therefore augment the simple algorithm by allowing it, 
with small probability (e.g., say $\sqrt{\e}$), to sample a second neighbor $u_2$ for each arriving vertex $v$, 
again with probabilities proportional to $z'_{u_2}$:
If the first sampled choice, $u_1$, is free, we match $v$ to $u_1$. Otherwise, if the second choice, $u_2$, 
is free, we match $v$ to $u_2$. 
What is the marginal probability that such an approach matches an incoming vertex $v$ to a given neighbor $u$? Letting $F_u$ denote the event that  $u$ is free when $v$ arrives, this probability is precisely
\begin{equation}\label{eq:03hg329hg}
  \Pr[F_u] \cdot \left(z'_{u} + z'_u\cdot \sqrt{\e} \cdot \sum_{w} z'_w \cdot (1-\Pr[F_w \mid F_u])  \right).
\end{equation}
Here the first term in the parentheses corresponds to the probability that $v$ matches to $u$ via the first choice, 
and the second term corresponds to the same happening via the second choice (which is only taken when the first choice 
fails). 

Ideally, we would like \eqref{eq:03hg329hg} to be at least $x_{uv}$ for all edges, which would imply a lossless rounding.
However, as mentioned earlier, this is difficult and in general impossible to do, even in much more restricted settings including one-sided bipartite vertex arrivals. We therefore settle for showing  that~\eqref{eq:03hg329hg} is at least $x_{uv} = \Pr[F_u]\cdot z_u$ for \emph{most} edges (weighted by $x_{uv}$). Even this goal, however, is challenging and requires a nontrivial understanding of the correlation structure of the random events $F_u$. To see this, note that for example if the $F_w$ events are perfectly positively correlated, i.e., $\Pr[F_w \mid F_u]=1$, then the possibility of picking $e$ as a second edge does not increase this edge's probability of being matched \emph{at all} compared to if we only picked a single edge per vertex. This results in $e$ being matched with probability $\Pr[F_u]\cdot z'_u = \Pr[F_u]\cdot z_u / \sum_w z_w = x_{uv} / \sum_w z_w$, which does not lead to any gain over the $\nicefrac1{2}$ competitive ratio of greedy. 
Such problems are easily shown not to arise if all $F_u$ variables are independent or negatively correlated. Unfortunately, positive correlation does arise from this process, and so we the need to control these positive correlations.

The core of our analysis is therefore dedicated to showing that even though positive correlations do arise, they are by and large rather weak. Our main technical contribution consists of developing techniques for bounding such positive correlations. The idea behind the analysis is to consider the primary choices and secondary choices of vertices as defining a graph, and showing that after a natural pruning operation that reflects the structure of dependencies, most vertices are most often part of a very small connected component in the graph. The fact that connected components are typically very small is exactly what makes positive correlations weak and results in the required lower bound on~\eqref{eq:03hg329hg} for most edges (in terms of $x$-value), which in turn yields our $\nicefrac{1}{2}+\Omega(1)$ competitive ratio.
\section{Edge Arrivals}\label{sec:edge}

In this section we prove the asymptotic optimality of the greedy algorithm for online matching under adversarial edge arrivals.
As discussed briefly in \Cref{sec:intro}, our main idea will be to provide a ``prefix hardness'' instance, where an underlying input and the arrival order is known to the online matching algorithm, but the prefix of the input to arrive (or ``termination time'') is not. Consequently, the algorithm must accrue high enough value up to each arrival time, to guarantee a high competitive ratio at all points in time. As we show, the fractional matching constraints rule out a competitive ratio of $\nicefrac{1}{2}+\Omega(1)$ even in this model where the underlying graph is known.

\begin{thm}\label{thm:hardness}
There exists an infinite family of bipartite graphs with maximum degree $n$ and edge arrival order for which any online matching algorithm is at best $\left(\frac{1}{2}+\frac{1}{2n+2}\right)$-competitive.
\end{thm}
\begin{proof}
We will provide a family of graphs for which no fractional online matching algorithm has better competitive ratio. Since any randomized algorithm induces a fractional matching algorithm, this immediately implies our claim.
The $n^{th}$ graph of the family, $G_n=(U,V,E)$, consists of a bipartite graph with $|U|=|V|=n$ vertices on either side. We denote by $u_i\in U$ and $v_i\in V$ the $i^{th}$ node on the left and right side of $G_n$, respectively. Edges are revealed in $n$ discrete rounds. In round $i=1,2,\dots,n$, the edges of a perfect matching between the first $i$ left and right vertices arrive in some order. I.e., a matching of $u_1,u_2,\dots,u_i$ and $v_1,v_2,\dots,v_i$ is revealed. Specifically, edges $(u_j,v_{i-j+1})$ for all $i\geq j$ arrive. (See \Cref{fig:hard-edge-instance} for example.) Intuitively, the difficulty for an algorithm attempting to assign much value to edges of $OPT$ is that the (unique) maximum matching $OPT$ changes every round, and no edge ever re-enters $OPT$.

\newcount\mycount
\newcommand{\DrawRound}[2]{% n, current round (k)
	\begin{tikzpicture}[scale=0.66]
    \node[draw, color=white] at (1,-0.25) {\textcolor{black}{$U$}};
    \node[draw, color=white] at (3,-0.25) {\textcolor{black}{$V$}};
    
    \foreach \number in {1,...,#1}
    \node[draw,circle,inner sep=0.01cm] (u-\number) at (1,-\number) {$u_\number$};
    \foreach \number in {1,...,#1}
    \node[draw,circle,inner sep=0.01cm] (v-\number) at (3,-\number) {$v_\number$};
    % edges
    
    \foreach \number in {1,...,#2}{
    	\foreach \i in {1,...,\number}{
    		\mycount=\number
    		\advance\mycount by -\i
    		\advance\mycount by 1
    		\draw [-, dashed] (u-\i) -- (v-\the\mycount);
    	}
    }
    \foreach \number in {#2}{
    	\foreach \i in {1,...,\number}{
    		\mycount=\number
    		\advance\mycount by -\i
    		\advance\mycount by 1
    		\draw [-, solid, line width=0.75pt] (u-\i) -- (v-\the\mycount);
    	}
    }
	\end{tikzpicture}
}

\newcommand{\subfigRound}[3] {% #1 = n and #2 = round number, #3 = subfig-width
	\begin{subfigure}[t]{#3}
	\begin{center}
	\DrawRound{#1}{#2}
	\subcaption{round #2}
	\end{center}
	\end{subfigure}
}

\def\n{5}
\begin{figure}[h]
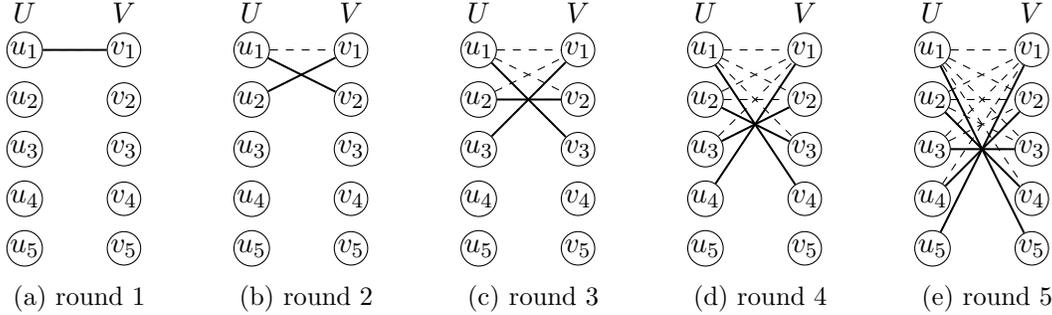

\begin{center}
\foreach \j in {1,...,\n}
{\subfigRound{\n}{\j}{0.175\textwidth}}
	\caption{$G_\n$, together with arrival order. Edges of  current (prior) round are solid (dashed).}
	\label{fig:hard-edge-instance}
\end{center}
\end{figure}
\let\n\relax

Consider some  $\alpha$-competitive fractional algorithm $\mathcal{A}$. We call the edge of a vertex $w$ in the (unique) maximum matching of the subgraph of $G_n$ following round $i$ the $i^{th}$ edge of $w$.
For $i\geq j$, denote by $x_{i,j}$ the value $\mathcal{A}$ assigns to the $i^{th}$ edge of vertex $u_j$ (and of $v_{i-j+1}$); i.e., to $(u_j,v_{i-j+1})$. 
By feasibility of the fractional matching output by $\mathcal{A}$, we immediately have that $x_{i,j}\geq 0$ for all $i,j$, as well as the following matching constraints for $u_j$ and $v_j$. (For the latter, note that the $i^{th}$ edge of $v_{i-j+1}$ is assigned value $x_{i,j}=x_{i,i-(i-j+1)+1}$ and so the $i^{th}$ edge of $v_j$ is assigned value $x_{i,i-j+1}$).
\begin{align}
    %(\sum_{e\ni u_k} x_e =) 
    \sum_{i=j}^{n} x_{i,j}\leq 1. & \qquad  \text{($u_j$ matching constraint)}\label{match-const-left} \\
    %(\sum_{e\ni v_k} x_e =) 
    \sum_{i=j}^{n} x_{i,i-j+1}\leq 1. & \qquad \text{($v_j$ matching constraint)}\label{match-const-right}
\end{align}

On the other hand, as $\mathcal{A}$ is $\alpha$-competitive, we have that after some $k^{th}$ round -- when the maximum matching has cardinality $k$ -- algorithm $\mathcal{A}$'s fractional matching  must have value at least $\alpha\cdot k$. (Else an adversary can stop the input after this round, leaving $\mathcal{A}$ with a worse than $\alpha$-competitive matching.) Consequently, we have the following competitiveness constraints.
\begin{equation}\label{comp-const}
\sum_{i=1}^k \sum_{j=1}^i x_{i,j} \geq \alpha \cdot k \qquad \forall k\in [n].
\end{equation}

Combining constraints \eqref{match-const-left}, \eqref{match-const-right} and \eqref{comp-const} together with the non-negativity of the $x_{i,k}$ yields the following linear program, LP($n$), whose optimal value upper bounds any fractional online matching algorithm's competitiveness on $G_n$, by the above.

\begin{figure}[h]
	\begin{center}
		\begin{tabular}{rll}
			maximize & $\alpha$ \\
			subject to: & $\sum_{i=j}^n x_{i,j}\leq 1$ & $\forall j\in [n]$ \\
			& $\sum_{i=j}^n x_{i,i-j+1}\leq 1$ & $\forall j\in [n]$ \\
			& $\sum_{i=1}^k \sum_{j=1}^i x_{i,j} \geq \alpha \cdot k$ & $\forall k\in [n]$ \\
			& $x_{i,j}\geq 0$ & $\forall i,j\in [n]$.
		\end{tabular}
	\end{center}
	\vspace{-0.5cm}
\end{figure} 

To bound the optimal value of LP($n$), we provide a feasible solution its LP dual, which we denote by Dual($n$). By weak duality, any dual feasible solution's value upper bounds the optimal value of LP($n$), which in turn upper bounds the optimal competitive ratio. Using the dual variables $\ell_j, r_j$ for the degree constraints of the $j^{th}$ left and right vertices respectively ($u_j$ and $v_j$) and dual variable $c_k$ for the competitiveness constraint of the $k^{th}$ round, we get the following dual linear program. Recall here again that $x_{i,i-j+1}$ appears in the matching constraint of $v_j$, with dual variable $r_j$, and so $x_{i,j}=x_{i,i-(i-j+1)+1}$ appears in the same constraint for $v_{i-j+1}$.)
\begin{figure}[h]
	\begin{center}
		\begin{tabular}{rll}
			minimize & $\sum_{j=1}^n \left( \ell_j+r_j \right)$ \\
			subject to: & $\sum_{k=1}^n k\cdot c_k \geq 1$  \\
			& $ \ell_j + r_{i-j+1} - \sum_{k=i}^n  c_k \geq 0$ & $\forall i\in [n], j \in [i]$ \\
			& $\ell_j, r_j, c_k \geq 0$ & $\forall j,k\in [n]$.
		\end{tabular}
	\end{center}
	\vspace{-0.3cm}
\end{figure} 

We provide the following dual  solution.
\begin{align*}
c_k &= \frac{2}{n(n+1)} \qquad \forall k \in [n]\\
\ell_j = r_j  & = \begin{cases} 
\frac{n-2(j-1)}{n(n+1)} & \mbox{if } j \leq n/2 +1 \\
0 & \mbox{if } n/2 +1 < j \leq n. \end{cases}
\end{align*}

We start by proving feasibility of this solution. 
The first constraint is satisfied with equality. For the second constraint, as $\sum_{k=i}^n  c_k = \frac{2 (n-i+1)}{n(n+1)}$ it suffices to show that 
$\ell_j + r_{i-j+1} \geq \frac{2 (n-i+1)}{n(n+1)}$ for all $i\in [n], j \in [i]$.
Note that if $j > n/2 +1$, then $\ell_j = r_j = 0 > \frac{n-2(j-1)}{n(n+1)}$. So, for all $j$ we have $\ell_j = r_j \geq \frac{n-2(j-1)}{n(n+1)}$.
Consequently, $\ell_j + r_{i-j+1} \geq \frac{n-2(j-1)}{n(n+1)} +\frac{n-2(i-j+1-1)}{n(n+1)} =  \frac{2 (n - i + 1)}{n (n+1)}$ for all $i\in [n], j \in [i]$. Non-negativity of the $\ell_j,r_j,c_k$ variables is trivial, and so we conclude that the above is a feasible dual solution.

It remains to calculate this dual feasible solution's value. We do so for $n$ even,\footnote{The case of $n$ odd is similar. As it is unnecessary to establish the result of this theorem, we omit it.} for which
\begin{align*}
        \sum_{j=1}^n (\ell_j + r_j ) = 2\cdot \sum_{j=1}^n \ell_j = 
        2\cdot \sum_{j=1}^{n/2 +1} \frac{n - 2(j-1)}{n(n+1)}  = \frac{1}{2} + \frac{1}{2n +2},
\end{align*}
completing the proof.
\end{proof}

{\textbf{Remark 1.}}  Recall that \citet{buchbinder2018online} and \citet{lee2017maximum} presented better-than-$\nicefrac{1}{2}$-competitive algorithms for bounded-degree graphs and bounded number of arrival batches. Our upper bound above shows that a deterioration of the competitive guarantees as the maximum degree and number of arrival batches increase (as in the algorithms of \cite{buchbinder2018online,lee2017maximum}) is inevitable.

\textbf{Remark 2.} Recall that the \emph{asymptotic} competitive ratio of an algorithm is the maximum $c$ such that the algorithm always guarantees value at least $ALG\geq c\cdot OPT - b$ for some fixed $b>0$. 
Our proof extends to this weaker notion of competitiveness easily, by revealing multiple copies of the hard family of \Cref{thm:hardness} and letting $x_{ik}$ denote the average of its 
counterparts over all copies.
%, and noting that the (linear) matching and competitiveness constraints must also hold on average, we also upper bound the asymptotic competitive ratio of any algorithm by $\nicefrac{1}{2}$.

\section{General Vertex Arrivals}\label{sec:vertex}

In this section we present a $(\nicefrac{1}{2}+\Omega(1))$-competitive randomized algorithm for online matching under general arrivals.
As discussed in the introduction, our approach will be to round (online) a \emph{fractional} online matching algorithm's output. Specifically, this will be an algorithm from the family of fractional algorithms introduced in \cite{wang2015two}. In \Cref{sec:fractional-algo} we describe this family of algorithms. To motivate our rounding approach, in \Cref{sec:warmup} we first present a simple lossless rounding method for a $\nicefrac{1}{2}$-competitive algorithm in this family. In \Cref{sec:imrpoved-algorithm} we then describe our rounding algorithm for a better-than-$\nicefrac{1}{2}$-competitive algorithm in this family. Finally, in \Cref{sec:analysis} we analyze this rounding scheme, and show that it yields a $(\nicefrac{1}{2}+\Omega(1))$-competitive algorithm.

\subsection{Finding a fractional solution}\label{sec:fractional-algo}
In this section we revisit the algorithm of \citet{wang2015two}, which beats the $\nicefrac{1}{2}$ competitiveness barrier for online fractional matching under general vertex arrivals. 
Their algorithm (technically, family of algorithms) applies the primal-dual method to compute both a fractional matching and a fractional vertex cover -- the dual of the fractional matching relaxation. The LPs defining these dual problems are as follows.

\begin{minipage}[t]{0.49\linewidth}

	\begin{center}
    \text{Primal-Matching}\par\medskip
		\begin{tabular}{rll}
			maximize & $\sum_{e \in E} x_e$ \\
			subject to: & $\sum_{u \in N(v)} x_{uv}\leq 1$ & $\forall u\in V$ \\
			            & $ x_e \geq 0 $ & $ \forall e \in E$    \\ 
		\end{tabular}
	\end{center}
%	\vspace{-0.5cm}
    \label{fig:primal-matching}	
\end{minipage}
\begin{minipage}[t]{0.49\linewidth}
	\begin{center}
    \text{Dual-Vertex Cover}\par\medskip
		\begin{tabular}{rll}
			minimize & $\sum_{u \in V} y_u$ \\
			subject to: & $  y_u + y_v  \geq 1$ & $\forall e = \{u,v\} \in E$ \\
			            & $ y_u \geq 0 $ & $ \forall u \in V$    \\ 
		\end{tabular}
	\end{center}
	\vspace{-0.5cm}
    \label{fig:dual-vertexcover}
\end{minipage}

Before introducing the algorithm of \cite{wang2015two}, we begin by defining the fractional online vertex cover problem for vertex arrivals. 
When a vertex $v$ arrives, if $N_v(v)$ denotes the previously-arrived neighbors of $v$, then for each $u\in N_v(v)$, a new constraint $y_u+y_v\geq 1$ is revealed, which an online algorithm should satisfy by possibly increasing $y_u$ or $y_v$.
Suppose $v$ has its dual value set to $y_v = 1-\theta$. Then all of its neighbors should have their dual increased to at least $\theta$. Indeed, an algorithm may as well increase $y_u$ to $\max\{y_u,\theta\}$. 
The choice of $\theta$ therefore determines an online fractional vertex cover algorithm. 
The increase of potential due to the newly-arrived vertex $v$ is thus $1-\theta + \sum_{u\in N_v(v)} (\theta-y_u)^+$.\footnote{Here and throughout the paper, we let $x^+ := \max\{0,x\}$ for all $x\in \mathbb{R}$.} In \cite{wang2015two} $\theta$ is chosen to upper bound this term by $1-\theta+f(\theta)$ for some function $f(\cdot)$. The primal solution (fractional matching) assigns values $x_{uv}$ so as to guarantee feasibility of $\vec{x}$ and a ratio of $\beta$ between the primal and dual values of $\vec{x}$ and $\vec{y}$, implying $\frac{1}{\beta}$-competitiveness of this online  fractional matching algorithm, by feasibility of $\vec{y}$ and weak duality. The algorithm, parameterized by a function $f(\cdot)$ and parameter $\beta$ to be discussed below, is given formally in \Cref{algo:TOBVC}. In the subsequent discussion, $N_v(u)$ denotes the set of neighbors of $u$ that arrive before $v$. 

\begin{algorithm}[h]
  \DontPrintSemicolon
  \SetKwInOut{Input}{Input}
  \SetKwInOut{Output}{Output}
  \Input{A stream of vertices $v_1, v_2, \dots v_n$. At step $i$, vertex $v_i$ and $N_{v_i}(v_i)$ are revealed.}
  \Output{A fractional vertex cover solution $\vec{y}$ and a fractional matching $\vec{x}$.}
  Let $y_u \gets 0$ for all $u$, let $x_{uv} \gets 0$ for all $u, v$. \;
  \ForEach{$v$ in the stream}{ 
    {
        \begin{tabular}{rll}
			&maximize $\theta$ \\
			&subject to: &$ \theta \leq 1$ \\
			           & & $ \sum_{u \in N_v(v)} \left( \theta - y_u \right)^{+} \leq f(\theta) $
		\end{tabular}
    } \;
    \ForEach{$u \in N_v (v) $}{
        $x_{uv} \xleftarrow[]{}  
                \frac{\left( \theta - y_u \right)^{+} }{\beta}
                \left(  1 + \frac{1-\theta}{f(\theta)}  \right).$ \;
         $y_u \xleftarrow[]{} \max \{ y_u, \theta\}$.
    }
    $y_v \xleftarrow[]{}  1 - \theta$. \;
    }
    
    \caption{Online general vertex arrival fractional matching and vertex cover}
    \label{algo:TOBVC}
\end{algorithm}

\Cref{algo:TOBVC} is parameterized by a function $f$ and a constant $\beta$. The family of functions considered by \cite{wang2015two} are as follows.
\begin{Def}
	Let $f_{\kappa}(\theta) := \left( \frac{1+ \kappa}{2}  - \theta  \right)^{\frac{1 + \kappa}{2 \kappa}} 
	\left( \theta +  \frac{\kappa - 1}{2}\right)^{\frac{\kappa - 1}{2 \kappa}}$.
	We define $W := \{ f_{\kappa} \mid \kappa \geq 1 \}$.
\end{Def}

As we will see, choices of $\beta$ guaranteeing feasibility of $\vec{x}$ are related to the following quantity.
\begin{Def}
For a given $f : \left[ 0,1 \right]  \xrightarrow{}  \mathbb{R}_{+} $    let $\beta^*(f) := \max_{\theta \in \left[ 0,1 \right]} 1 +  f(1-\theta)+  \int_{\theta}^{1} \frac{1-t}{f(\theta)} \,d\theta$.
\end{Def}

For functions $f\in W$ this definition of $\beta^*(f)$ can be simplified to $\beta^*(f)=1+f(0)$, due to the observation  (see \cite[Lemmas 4,5]{wang2015two}) that all functions $f\in W$ satisfy 
\begin{equation}\label{WW-tight}
	\beta^*(f) = 1+f(1-\theta)+\int_{\theta}^1 \frac{1-\theta}{f(\theta)}\,d\theta \qquad \forall \theta\in [0,1].
\end{equation}

As mentioned above, the competitiveness of \Cref{algo:TOBVC} for appropriate choices of $f$ and $\beta$ is obtained by relating the overall primal and dual values, $\sum_e x_e$ and $\sum_v y_v$. As we show (and rely on later), one can even bound individual vertices' contributions to these sums. In particular, for any vertex $v$'s arrival time, each vertex $u$'s contribution to $\sum_e x_e$, which we refer to as its \emph{fractional degree}, $x_u := \sum_{w\in N_v(u)} x_{uw}$, can be bounded in terms of its dual value by this point, $y_u$, as follows.

\begin{restatable}{lem}{refinedwwbound}\label{refined-ww-xu-bound}
For any vertex $u,v \in V $, let $y_u$ be the potential of $u$ prior to arrival of $v$. Then the fractional degree just before $v$ arrives,  $x_u:=\sum_{w\in N_v(u)} x_{uw}$, is bounded as follows:
\begin{align*}
\frac{y_u}{\beta} \leq	x_u  \leq \frac{y_u + f(1-y_u)}{\beta}.
\end{align*}
\end{restatable}
Broadly, the lower bound on $x_u$ is obtained by lower bounding the increase $x_u$ by the increase to $y_u/\beta$ after each vertex arrival, while the upper bound follows from a simplification of a bound given in \cite[Invariant 1]{wang2015two} (implying feasibility of the primal solution), which we simplify using \eqref{WW-tight}. See \Cref{app:fractional-algo} for a full proof.

Another observation we will need regarding the functions $f\in W$ is that they are decreasing.
\begin{obs}\label{f-non-increasing}
	Every function $f\in W$ is non-increasing in its argument in the range $[0,1]$.
\end{obs}
\begin{proof}
	As observed in \cite{wang2015two}, differentiating \eqref{WW-tight} with respect to $z$ yields $- f'(1-z) - \frac{1-z}{f(z)} = 0$, from which we obtain $f(z)\cdot f'(1-z) = z-1$. Replacing $z$ by $1-z$, we get $f(1-z)\cdot f'(z) = -z$, or $f'(z) = -\frac{z}{f(1-z)}$. As $f(z)$ is positive for all $z\in [0,1]$, we have that $f'(z)<0$ for all $z\in [0,1]$.
\end{proof}

The next lemma of \cite{wang2015two} characterizes the achievable competitiveness of \Cref{algo:TOBVC}.

\begin{lem}[\cite{wang2015two}]\label{ww-approx}
\cref{algo:TOBVC} with function $f \in W$ and $\beta \geq \beta^{*}(f) = 1+f(0)$ is $\frac{1}{\beta}$ competitive.
\end{lem}
\citet{wang2015two} showed that taking $\kappa \approx 1.1997$ and $\beta = \beta^*(f_\kappa)$, \Cref{algo:TOBVC} is $\approx 0.526$ competitive. In later sections we show how to round the output of \Cref{algo:TOBVC} with $f_\kappa$ with $\kappa=1+ 2 \epsilon$ for some small constant $\epsilon$ and $\beta=2-\epsilon$ to obtain a $(\nicefrac{1}{2}+\Omega(1))$-competitive algorithm. But first, as a warm up, we show how to round this algorithm with $\kappa=1$ and $\beta=\beta^*(f_1)=2$.

%!TEX root = ./main.tex
\subsection{Warmup: a \texorpdfstring{$\nicefrac{1}{2}$}{(1/2)}-competitive randomized algorithm}\label{sec:warmup}

In this section we will round the $\nicefrac{1}{2}$-competitive fractional algorithm obtained by running \cref{algo:TOBVC} with function $f(\theta) = f_1 (\theta) = 1-\theta$ and $\beta = \beta^{*} (f) = 2$. We will devise a lossless rounding of this fractional matching algorithm, by including each edge $e$ in the final matching with a probability equal to the fractional value $x_e$ assigned to it by \cref{algo:TOBVC}. 
Note that if $v$ arrives after $u$, then if $F_u$ denotes the event that $u$ is free when $v$ arrives, then edge $\{u,v\}$ is matched by an online algorithm with probability $\Pr[\{u,v\}\in M] = \Pr[\{u,v\}\in M \mid F_u] \cdot \Pr[F_u]$. Therefore, to match each edge $\{u,v\}$ with probability $x_{uv}$, we need $\Pr[\{u,v\} \in M \mid F_u] = x_{uv} / \Pr[F_u]$.
That is, we must match $\{u,v\}$ with probability $z_u = x_{uv}/\Pr[F_u]$ conditioned on $u$ being free. The simplest way of doing so (if possible) is to pick an edge $\{u,v\}$ with the above probability $z_u$ always, and to match it only if $u$ is free.
\Cref{algo:warmup_algorithm} below does just this, achieving a lossless rounding of this fractional algorithm.
As before, $N_v(u)$ denotes the set of neighbors of $u$ that arrive before $v$. 

\begin{algorithm}[h]
  \DontPrintSemicolon
  \SetKwInOut{Input}{Input}
  \SetKwInOut{Output}{Output}
  \Input{A stream of vertices $v_1, v_2, \dots , v_n$. At step $i$, vertex $v_i$ and $N_{v_i}(v_i)$ are revealed.}
  \Output{A matching $M$.}
  Let $y_u \gets 0$ for all $u$, let $x_{uv} \gets 0$ for all $u, v$. \;
  Let $M \gets \emptyset$. \;
  \ForEach{$v$ in the stream}{ 
    Update $y_u$'s and $x_{uv}$'s using~\Cref{algo:TOBVC} with $\beta = 2$ and $f = f_{1}$. \;
    \ForEach{$u \in N_v (v) $}{
      $z_u \gets \frac{x_{u v}}{\Pr[\text{$u$ is free when $v$ arrives}]}.$ \tcp*{$z_u$ is $x_{uv} / (1 - y_u)$ as shown later} \label{line:zu}
        }
    Sample (at most) one neighbor $u\in N_v(v)$ according to $z_u$. \label{line:sample}\;
    \If{a free neighbor $u$ is sampled}{
        Add $\{u,v\}$ to $M$.
    }
  }    
    \caption{Online vertex arrival warmup randomized fractional matching}
    \label{algo:warmup_algorithm}
\end{algorithm}

\Cref{algo:warmup_algorithm} is well defined if for each vertex $v$'s arrival, $z$ is a probability distribution; i.e., $\sum_{u\in N_v(v)} z_u\leq 1$. The following lemma asserts precisely that. Moreover, it asserts that \Cref{algo:warmup_algorithm} matches each edge with the desired probability.

\begin{lem}\label{lem:z-distribution}
\Cref{algo:warmup_algorithm} is well defined, since for every vertex $v$ on arrival, $z$ is a valid probability distribution. Moreover, for each $v$ and $u\in N_v(v)$, it matches edge $\{u,v\}$ with probability $x_e$.
\end{lem}
\begin{proof}
	We prove both claims in tandem for each $v$, by induction on the number of arrivals. 
	For the base case ($v$ is the first arrival), the set $N_v(v)$ is empty and thus both claims are trivial. Consider the arrival of a later vertex $v$.
	By the inductive hypothesis we have that each vertex $u\in N_v(v)$ is previously matched with probability $\sum_{w\in N_v(u)} x_{wu}$.
	But by our choice of $f(\theta)=f_1(\theta)=1-\theta$ and $\beta=2$, if $w$ arrives after $u$, then $y_u$ and $\theta$ at arrival of $w$ satisfy $x_{uw} = \frac{(\theta - y_u)^{+}}{\beta}\cdot \left(1+\frac{1-\theta}{f(\theta)}\right) = (\theta - y_u)^{+}$. 
	That is, $x_{uw}$ is precisely the increase in $y_u$ following arrival of $w$.
	On the other hand, when $u$ arrived we have that its dual value $y_u$ increased by $1-\theta = \sum_{v'\in N_u(u)} (\theta - y_{v'})^+ = \sum_{v'\in N_u(u)} x_{uv'}$. 
	To see this last step, we recall first that by definition of \Cref{algo:TOBVC} and our choice of $f(\theta)=1-\theta$, the value $\theta$ on arrival of $v$ is chosen to be the largest $\theta\leq 1$ satisfying  
	\begin{align}\label{eq:theta-constraint}
		\sum_{\forall u \in N_v(v)} (\theta - y_u)^{+}\leq 1-\theta.
	\end{align}
	But the inequality \eqref{eq:theta-constraint} is an equality whether or not $\theta=1$ (if $\theta=1$, both sides are zero).
	We conclude that $y_u = \sum_{v'\in N_v(u)} x_{uv'}$ just prior to arrival of $v$. 
	But then, by the inductive hypothesis, this implies that $\Pr[\mbox{$u$ free when $v$ arrives}] = 1-y_u$ (yielding an easily-computable formula for $z_u$).
	Consequently, by \eqref{eq:theta-constraint} we have that when $v$ arrives $z$ is a probability distribution, as
	\begin{align*}
	\sum_{u\in N_v(v)} z_u & = \sum_{u\in N_v(v)} \frac{(\theta - y_u)^{+}}{1-y_u} \leq \sum_{u\in N_v(v):\, y_u\leq \theta} \frac{(\theta - y_u)^{+}}{1-\theta} = \sum_{u\in N_v(v)} \frac{(\theta - y_u)^{+}}{1-\theta} \leq 1.
	\end{align*}
	
	Finally, for $u$ to be matched to a latter-arriving neighbor $v$, it must be picked and free when $v$ arrives, and so $\{u,v\}$ is indeed matched with probability
	\begin{align*}
	\Pr[\{u,v\}\in M] & = \frac{x_{uv}}{\Pr[\mbox{$u$ is free when $v$ arrives}]} \cdot \Pr[\mbox{$u$ is free when $v$ arrives}] = x_{uv}.\qedhere
	\end{align*}
\end{proof}

In the next section we present an algorithm which allows to round better-than-$\nicefrac{1}{2}$-competitive algorithms derived from \Cref{algo:TOBVC}.

%!TEX root = ./main.tex
\subsection{An improved algorithm}\label{sec:imrpoved-algorithm}

In this section, we build on \Cref{algo:warmup_algorithm} and show how to improve it to get a 
$(1/2 + \Omega(1))$ competitive ratio.

There are two concerns when modifying \Cref{algo:warmup_algorithm} to work for a general
function from the family $W$.
The first is how to compute the probability that a vertex $u$ is free when vertex $v$ arrives, 
in \Cref{line:zu}.
In the simpler version, we inductively showed that this probability is simply $1 - y_u$, where 
$y_u$ is the dual value of $u$ as of $v$'s arrival (see the proof of \Cref{lem:z-distribution}). 
With a general function $f$, this probability is no longer given by a simple formula. 
Nevertheless, it is easily fixable: We can either use Monte Carlo sampling to estimate the 
probability of $u$ being free at $v$'s arrival to a given inverse polynomial accuracy, 
or we can in fact exactly compute these probabilities by maintaining their marginal values 
as the algorithm progresses. 
In what follows, we therefore assume that our algorithm can compute these probabilities 
exactly.

The second and more important issue is with the sampling step in \Cref{line:sample}.
In the simpler algorithm, this step is well-defined as the sampling probabilities indeed form a 
valid distribution: 
I.e., $\sum_{u \in N_v(v)}z_u \leq 1$ for all vertices $v$.
However, with a general function $f$, this sum can exceed one, rendering the sampling step in 
\Cref{line:sample} impossible.
Intuitively, we can normalize the probabilities to make it a proper distribution, but by doing so,
we end up losing some amount from the approximation guarantee.
We hope to recover this loss using a second sampling step, as we mentioned in \Cref{sec:techniques}
and elaborate below.

Suppose that, instead of $\beta = 2$ and $f = f_1$ (i.e., the function $f(\theta) = 1 - \theta$), 
we use $f = f_{1 + 2 \e}$ and $\beta = 2 - \e$ to define $x_{uv}$ and $y_u$ values.
As we show later in this section, for an $\e$ sufficiently small, we then have $\sum_{u \in N_v(v)} 
z_u \leq 1 + O(\e)$, implying that the normalization factor is at most $1 + O(\e)$.
However, since the approximation factor of the fractional solution is only $1/2 + O(\e)$ for such 
a solution, (i.e., $\sum_{\{u,v\} \in E} x_{uv} \geq (1/\beta) \cdot \sum_{u \in V} y_u$), the 
loss due to normalization is too significant to ignore.

Now suppose that we allow arriving vertices to sample a second edge with a small 
(i.e., $\sqrt{\e}$) probability and match that second edge if the endpoint of the first sampled 
edge is already matched.
Consider the arrival of a fixed vertex $v$ such that $\sum_{u \in N_v(v)}z_u > 1$, and let 
$z'_u$ denote the normalized $z_u$ values. Further let $F_w$ denote the event that vertex $w$ is free (i.e, unmatched) at the arrival of $v$. 
Then the probability that $v$ matches $u$ for some $u \in N_v(v)$ using either of the two sampled
edges is
\begin{align}
    \Pr[F_u] \cdot \left(z'_{u} + z'_u \sqrt{\e} \cdot 
    \sum_{w\in N_v(v)} z'_w \cdot (1-\Pr[F_w \mid F_u]) \right), \label{prob-match}
\end{align}
which is the same expression from \eqref{eq:03hg329hg} from \Cref{sec:techniques}, restated here for quick reference. 
Recall that the first term inside the parentheses accounts for the probability that $v$ matches 
$u$ via the first sampled edges, and the second term accounts for the probability that the same 
happens via the second sampled edge. 
Note that the second sampled edge is used only when the first one
is incident to an already matched vertex and the other endpoint of the second edge is free. 
Hence we have the summation of conditional probabilities in the second term, 
where the events are conditioned on the other endpoint, $u$, being free.
If the probability given in \eqref{prob-match} is $x_{uv}$ for all $\{u,v\} \in E$, we would have 
the same guarantee as the fractional solution  $x_{uv}$, and the rounding would be lossless. 
This seems unlikely, yet we can show that the quantity in \eqref{prob-match} is at least  
$(1 - \e^2) \cdot x_{uv}$ for most (not by number, but by the total fractional value of $x_{uv}$'s) of 
the edges in the graph, showing that our rounding is \emph{almost} lossless. 
We postpone further discussion of the analysis to Section~\ref{sec:analysis} where we highlight the main ideas 
before proceeding with the formal proof.

\begin{algorithm}[h]
	\DontPrintSemicolon
	\SetKwInOut{Input}{Input}
	\SetKwInOut{Output}{Output}
	\Input{A stream of vertices $v_1, v_2, \dots, v_n$. At step $i$, vertex $v_i$ and $N_{v_i}(v)$ are revealed.}
	\Output{A matching $M$.}  
	Let $y_u \gets 0$ for all $u$, let $x_{uv} \gets 0$ for all $u, v$. \;
	Let $M \gets \emptyset$. \;
	\ForEach{$v$ in the stream}{ 
		Update $y_u$'s and $x_{uv}$'s using~\Cref{algo:TOBVC} with $\beta = 2 - \epsilon$ and $f = f_{1 + 2\epsilon}$. \label{line:fractional-to-round} \;
		\ForEach{$u \in N_v(v)$}{
          \tcc{Compute $\Pr[\text{$u$ is free when $v$ arrives}]$ as explained in \Cref{sec:imrpoved-algorithm}}
          $z_u \gets \frac{x_{uv}}{\Pr[\text{$u$ is free when $v$ arrives}]}$. \label{line:sample-improved-alg} 
		}
		\ForEach{$u \in N_v(v)$}{
			$z'_u \gets z_u / \max\left\{1,\sum_{u \in N_v(v)} z_u \right\}$.\label{line:normalize}
		}		
		Pick (at most) one $u_1 \in N_v(v)$ with probability $z'_{u_1}$. \label{line:first-pick} \;
		\If{$\sum_{u\in N_v(v)} z_u > 1$}{
			With probability $\sqrt{ \epsilon}$, pick (at most) one $u_2 \in N_v(v)$ with probability $z'_{u_2}$. \label{line:second-pick} \;
            \tcc{Probability of dropping edge $\{u, v\}$  can be computed using \eqref{prob-match}.}
			Drop $u_2$ with minimal probability ensuring $\{u_2,v\}$ is matched with probability at most $x_{u_2 v}$.\label{line:second-pick-drop} 
		}		
		\If{a free neighbor $u_1$ is sampled}{Add $\left\{ u_1, v \right\}$ to $M$.
		}
        \ElseIf{a free neighbor $u_2$ is sampled}{Add $\left\{ u_2, v \right\}$ to $M$.} 		
	\caption{A randomized online matching algorithm under general vertex arrivals.}
	\label{algo:wwrounding}  
}
\end{algorithm}

Our improved algorithm is outlined in~\cref{algo:wwrounding}. Up until \Cref{line:sample-improved-alg}, it is 
similar to \Cref{algo:warmup_algorithm} except that it uses $\beta = 2 - \e$ and $f = f_{1 + 2\e}$ where
we choose $\e > 0$ to be any constant small enough such that the results in the analysis hold.
In \Cref{line:normalize}, if the
sum of $z_u$'s exceeds one we normalize the $z_u$ to obtain a valid probability distribution $z'_u$.
In \Cref{line:first-pick}, we sample the first edge incident to an arriving vertex $v$.
In \Cref{line:second-pick}, we sample a second edge incident to the same vertex with probability $\sqrt{\e}$ 
if we had to scale down $z_u$'s in \Cref{line:normalize}.
Then in \Cref{line:second-pick-drop}, we drop the sampled second edge with the minimal probability
to ensure that no edge $\{u, v\}$ is matched with probability more than $x_{uv}$. 
Since \eqref{prob-match} gives the exact probability of $\{u, v\}$ being matched, this probability of 
dropping an edge $\{u, v\}$ can be computed by the algorithm.
However, to compute this, we need the conditional probabilities $\Pr[F_w \mid F_u]$, which again
can be estimated using Monte Carlo sampling\footnote{It is also possible to compute them exactly if we 
allow the algorithm to take exponential time.}.
In the subsequent lines, we match $v$ to a chosen free neighbor (if any) among its chosen neighbors, 
prioritizing its first choice.

For the purpose of analysis we view \Cref{algo:wwrounding} as constructing a greedy matching on a directed acyclic graph (DAG) $H_\tau$ defined in the following two definitions.

\begin{Def}[Non-adaptive selection graph $G_\tau$]\label{def:g-tau}
Let $\tau$ denote the random choices made by the vertices of $G$. Let $G_\tau$ be the DAG defined by all the arcs $(v, u_1)$, $(v, u_2)$ for all vertices $v \in V$. 
We call the arcs $(v, u_1)$ \emph{primary} arcs, and the arcs $(v, u_2)$ the \emph{secondary} arcs. 
\end{Def}

\begin{Def}[Pruned selection graph $H_\tau$]\label{def:h-tau}
Now construct $H_\tau$ from $G_\tau$ by removing all arcs $(v, u)$ (primary or secondary) such that there exists 
a \emph{primary} arc $(v', u)$ with $v'$ arriving before $v$. We further remove a secondary arc $(v,u)$ if there is a primary arc $(v,u)$; 
i.e., if a vertex $u$ has at least one incoming primary arc, remove all incoming primary arcs 
that came after the first primary arc and all secondary arcs that came after or from the same vertex as the first primary arc. 
\end{Def}

It is easy to see that the matching constructed by \Cref{algo:wwrounding} is a greedy matching constructed on $H_\tau$ based on order of arrival and prioritizing primary arcs. The following lemma shows that the set of matched vertices obtained by this greedy matching does not change much for any change in the random choices of a single vertex $v$, which will prove useful later on. 
It can be proven rather directly by an inductive argument showing the size of the symmetric difference in matched vertices in $G_\tau$ and $G_{\tau'}$ does not increase after each arrival besides the arrival of $v$, whose arrival clearly increases this symmetric difference by at most two. See \Cref{app:imrpoved-algorithm} for details.

\begin{restatable}{lem}{freevertices}\label{lem:freevertices}
  Let $G_{\tau}$ and $G_{\tau'}$ be two realizations of the random digraph where all the vertices in 
  the two graphs make the same choices except for \emph{one} vertex $v$. 
  Then the number of vertices that have different \emph{matched status} (free/matched) in the matchings computed in $H_{\tau}$ and $H_{\tau'}$
  at any point of time is at most two.
\end{restatable}

%!TEX root = ./main.tex
\subsection{Analysis}\label{sec:analysis}
\newcommand{\EventEQ}[1]{\ensuremath{\mathrm{EQ}_{#1}}}
\newcommand{\EventIN}[1]{\ensuremath{\mathrm{IN}_{#1}}}

In this section, we analyze the competitive ratio of~\Cref{algo:wwrounding}. We start with an outline of the analysis where we highlight the main ideas.

\subsubsection{High-Level Description of Analysis}
As described in~\Cref{sec:imrpoved-algorithm}, the main difference compared to the simpler $\nicefrac{1}{2}$-competitive algorithm is the change of the construction of the fractional solution, which in turn makes the rounding more complex. In particular, we may have at the arrival of a vertex $v$ that $\sum_{u\in N_v(v)} z_u > 1$.  The majority of the analysis is therefore devoted to such ``problematic'' vertices since otherwise, if $\sum_{u\in N_v(v)} z_u \leq 1$, the rounding is lossless due to the same reasons as described in the simpler setting of~\Cref{sec:warmup}. We now outline the main ideas in analyzing a vertex $v$ with $\sum_{u\in N_v(v)} z_u > 1$. Let $F_w$ be the event that vertex $w$ is free (i.e., unmatched) at the arrival of $v$. Then, as described in~\Cref{sec:imrpoved-algorithm}, the probability that we select edge $\{u,v\}$ in our matching is the minimum of $x_{uv}$  (because of the pruning in \Cref{line:second-pick-drop}), and
\begin{align*}
\Pr[F_u] \cdot \left(z'_{u} + z'_u \sqrt{\varepsilon} \cdot \sum_{w\in N_v(v)} z'_w \cdot (1-\Pr[F_w \mid F_u])  \right).
\end{align*}
 By definition, $\Pr[F_u] \cdot z_u = x_{uv}$, and the expression inside the parentheses is at least $z_u$ (implying $\Pr[\{u,v\}\in M]= x_{uv})$ if
\begin{align}
1 + \sqrt{\varepsilon} \cdot \sum_{w\in N_v(v)} z'_w \cdot (1-\Pr[F_w \mid F_u])  \geq \frac{z_u}{z'_{u}}.
\label{eq:enough}
\end{align}
To analyze this inequality, we first use the structure of the selected function
$f= f_{1+2\e}$ and the selection of $\beta = 2-\e$ to show that if $\sum_{u\in
N_v(v)} z_u > 1$ then several structural properties hold (see~\Cref{lem:propertiesWW} and~\Cref{cor:ww} in~\Cref{sec:propertiesWW}). In particular, there are absolute constants $0<c<1$ and $C>1$ (both independent of $\e$) such that 
\begin{enumerate}
  \item $\sum_{u\in
    N_v(v)} z_u \leq 1+ C\e$;
  \item $z_u \leq C\sqrt{\e}$ for every $u\in N_v(v)$; and
  \item $c \leq \Pr[F_w]  \leq 1-c$ for every $w\in N_{v}(v)$.
\end{enumerate}
The first property implies that the right-hand-side of~\eqref{eq:enough} is at most $1 + C\e$; and the second property implies that $v$ has at least $\Omega(1/\sqrt{\e})$ neighbors and that each neighbor $u$ satisfies $z'_u \leq z_u \leq C\sqrt{\e}$. 

For simplicity of notation, we assume further in the high-level  overview  that $v$ has exactly $1/\sqrt{\e}$ neighbors and each $u\in N_v(v)$ satisfies $z'_u = \sqrt{\e}$. Inequality~\eqref{eq:enough} would then be implied by
\begin{align}
 \sum_{w\in N_v(v)}   (1-\Pr[F_w \mid F_u])  \geq  C\,.
 \label{eq:enough_simple}
\end{align}
To get an intuition why we would expect the above inequality to hold, it is instructive to consider the unconditional version:
\begin{align*}
  \sum_{w\in N_v(v)}   (1-\Pr[F_w ]) \geq c |N_v(v)| = c/\sqrt{\eps} \gg C \,,
\end{align*}
where the first inequality is from the fact that  $\Pr[F_w] \leq 1-c$ for any neighbor $w\in N_v(v)$. The large slack  in the last inequality, obtained by selecting $\e>0$ to be a sufficiently small constant, is used to bound the impact of conditioning on the event $F_u$. Indeed, due to the large slack, we  have that~\eqref{eq:enough_simple} is satisfied if the quantity  $\sum_{w \in N_v(v)}\Pr[F_w|F_u]$ is not too  far away from the same summation with unconditional 
probabilities, i.e., $\sum_{w \in N_v(v)} \Pr[F_w]$.  Specifically, it is sufficient to show
\begin{align}
  \sum_{w \in N_v(v)}\left(\Pr[F_w|F_u] - \Pr[F_w]\right) \leq c/\sqrt{\eps} - C \,.
  \label{eq:simpleFsum}
\end{align}
We do so by bounding the correlation between the events $F_u$ and $F_w$ in a highly non-trivial manner, 
which constitutes the heart of our analysis.
The main challenges are that  events $F_u$  and $F_w$ can be positively correlated and that, by conditioning on $F_u$, the primary and secondary choices of different vertices are no longer independent. 

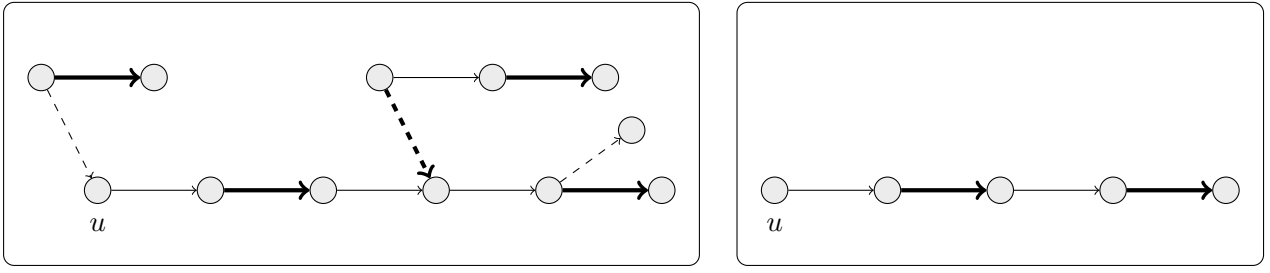
\begin{figure}[t]
  \centering
  \begin{tikzpicture}

         \draw[rounded corners] (-1.25, -1.0) rectangle (8.0, 2.5);
    \node[sgvertex] (u) at (0,0) {};
    \node[sgvertex] (u1) at (1.5,0) {};
    \node[sgvertex] (u2) at (3,0) {};
    \node[sgvertex] (u3) at (4.5,0) {};
    \node[sgvertex] (u4) at (6,0) {};
    \node[sgvertex] (u5) at (7.5,0) {};

%    \node[sgvertex] (a1) at (2.25,1.5) {};
    \node[sgvertex] (a2) at (3.75,1.5) {};
    \node[sgvertex] (a3) at (5.25,1.5) {};
    \node[sgvertex] (a4) at (6.75,1.5) {};
    \node[sgvertex] (b4) at (7.10,0.8) {};
    
    \node[sgvertex] (l1) at (-0.75,1.5) {};
    \node[sgvertex] (l2) at (0.75,1.5) {};

    \draw (u) edge[-> ] (u1)
          (u1) edge[->, ultra thick] (u2)
          (u2) edge[->] (u3)
          (u3) edge[->] (u4)
          (u4) edge[->, ultra thick] (u5) edge[->, dashed] (b4)
          (a2) edge[->, dashed, ultra thick] (u3)
          (a2) edge[->] (a3)
          (a3) edge[->, ultra thick] (a4)
          (l1) edge[->, dashed] (u) edge[->, ultra thick] (l2);

    \node[below =0.06cm of u] {$u$};
          
       \begin{scope}[xshift=9cm]
         \draw[rounded corners] (-0.5, -1.0) rectangle (6.5, 2.5);
    \node[sgvertex] (u) at (0,0) {};
    \node[sgvertex] (u1) at (1.5,0) {};
    \node[sgvertex] (u2) at (3,0) {};
    \node[sgvertex] (u3) at (4.5,0) {};
    \node[sgvertex] (u4) at (6,0) {};
    \draw (u) edge[->] (u1)
          (u1) edge[->, ultra thick] (u2)
          (u2) edge[->] (u3)
          (u3) edge[->, ultra thick] (u4);
    \node[below =0.06cm of u] {$u$};
      \end{scope}

  \end{tikzpicture}
  \caption{Two examples of the component of $H_\tau$ containing $u$. Vertices are depicted from right to left in the arrival order. Primary and secondary arcs are solid and dashed, respectively. The edges that take part in the matching are thick.}
  \label{fig:component}
\end{figure}
We overcome the last difficulty by replacing the conditioning on $F_u$ by a conditioning on the component in $H_\tau$ (at the time of $v$'s arrival) that includes $u$. As explained in~\Cref{sec:imrpoved-algorithm}, the matching output by our algorithm is equivalent to the greedy matching constructed in $H_\tau$ and so the component containing $u$ (at the time of $v$'s arrival) determines $F_u$. But how can this component look like, assuming the event $F_u$? First, $u$ cannot have any incoming primary arc since then $u$ would be matched (and so the event $F_u$ would be false). However, $u$ could have incoming secondary arcs, assuming that the tails of those arcs are matched using their primary arcs. Furthermore, $u$ can have an outgoing primary and possibly a secondary arc if the selected neighbors are already matched. These neighbors can in turn have incoming secondary arcs, at most one incoming primary arc (due to the pruning in the definition of $H_\tau$),  and outgoing primary and secondary arcs; and so on. In~\Cref{fig:component}, we give two examples  of the possible structure, when conditioning on $F_u$,  of $u$'s component in $H_\tau$  (at the time of  $v$'s arrival).
The left example contains secondary arcs, whereas the component on the right is arguably simpler and only contains primary arcs. 

An important step in our proof is to prove that, for most vertices $u$, the component is of the simple form depicted to the right  with probability almost one. That is, it is a path $P$ consisting of primary arcs, referred to as a primary path (see~\Cref{def:primarypath}) that further satisfies:
\begin{enumerate}[(i)]
  \item it has length $O(\ln(1/\e))$; and
  \item the total $z$-value of the arcs in the blocking set of $P$ is $O(\ln(1/\e))$. The blocking set is defined in~\Cref{def:blocking}. Informally, it contains those arcs that if appearing as primary arcs in $G_\tau$ would cause  arcs of $P$ to be pruned  (or blocked) from $H_\tau$. 
\end{enumerate}
Let $\mathcal{P}$ be the primary paths of above type that appear with positive probability as $u$'s component in $H_\tau$. Further let $\EventEQ{P}$ be the event that $u$'s component equals $P$. Then we show (for most vertices) that $\sum_{P \in \mathcal{P}} \Pr[\EventEQ{P} \mid F_u]$ is almost one. For simplicity, let us assume here that the sum is equal to one. Then by the law of total probability and since $\sum_{P \in \mathcal{P}} \Pr[\EventEQ{P} \mid F_u] =1$, 
\begin{align*}
  \sum_{w \in N_v(v)}\left(\Pr[F_w \mid F_u] - \Pr[F_w]\right) & = \sum_{P\in \mathcal{P}} \Pr[\EventEQ{P} \mid F_u]\left( \sum_{w \in N_v(v)}\left(\Pr[F_w \mid F_u, \EventEQ{P}] - \Pr[F_w]\right) \right) \\
  & = \sum_{P\in \mathcal{P}} \Pr[\EventEQ{P} \mid F_u]\left( \sum_{w \in N_v(v)}\left(\Pr[F_w \mid \EventEQ{P}] - \Pr[F_w]\right) \right),
\end{align*}
where the last equality is because the component $P$ determines $F_u$. The proof is then completed by analyzing the term inside the parentheses for each  primary path $P\in \mathcal{P}$ separately. 
As we prove in~\Cref{lm:conditional-product}, the independence of primary and secondary arc choices of vertices is maintained after conditioning on $\EventEQ{P}$.\footnote{To be precise, conditioning on a primary path $P$ with a so-called termination certificate $T$, see~\Cref{def:primarypath}. In the overview, we omit this detail and consider the event $\EventEQ{P,T}$ (instead of $\EventEQ{P}$) in the formal proof.} Furthermore, we show  that there is a bijection between the outcomes of the  unconditional and the conditional distributions, so that the expected number of vertices that make different choices under this pairing can be upper bounded  by roughly the length of the path plus the $z$-value of the edges in the blocking set. So, for a path $P$ as above, we have that the expected number of vertices that make different choices in the paired outcomes is $O(\ln(1/\e))$  which, by~\Cref{lem:freevertices}, implies that the expected number of vertices that change matched status is also upper bounded by $O(\ln(1/\e))$.  In other words,
we have for every $P\in \mathcal{P}$ that
\begin{align*}
  \sum_{w \in N_v(v)}\left(\Pr[F_w|\EventEQ{P}] - \Pr[F_w]\right) \leq \sum_{w \in V}\left(\Pr[F_w|\EventEQ{P}] - \Pr[F_w]\right) = O(\ln(1/\e)),
\end{align*}
which implies~\eqref{eq:simpleFsum} for a small enough choice of $\eps$. This completes the overview of the main steps in the analysis. The main difference in the formal proof is that not all vertices satisfy that their component is a short primary path  with probability close to $1$. To that end, we define the notion of \emph{good} vertices in~\Cref{sec:boundvar}, which are the vertices that are very unlikely 
to have long directed paths of primary arcs rooted at them. These are exactly the vertices $v$ for which we can perform the above analysis for most neighbors $u$ (in the proof of the ``key lemma'') implying that the rounding is almost lossless for $v$. Then, in~\Cref{sec:bad-is-ok}, we show using a rather simple charging scheme that most of the vertices in the graph are good. Finally, in~\Cref{sec:competitive}, we put everything together and prove~\Cref{thm:vertex-intro}.

\subsubsection{Useful Properties of \texorpdfstring{$W$}{W} Functions and \Cref{algo:wwrounding}}
\label{sec:propertiesWW}
For the choice of $f=f_{1+2\epsilon}$ as we choose, we have $f(\theta) = \left(1+ \varepsilon - \theta\right) \cdot \left(\frac{\theta+\varepsilon}{1+\varepsilon-\theta}
\right)^{\frac{\varepsilon}{1+2\varepsilon}}$. In \Cref{app:WW-functions} we give a more manageable upper bound for $f(\theta)$ which holds for sufficiently small $\epsilon$. Based on this simple upper bound on $f$ and some basic calculus, we obtain the following useful structural properties for the conditional probabilities, $z_u$, of \Cref{algo:wwrounding}. See \Cref{app:WW-functions}.

\begin{restatable}{lem}{WWproperties}(Basic bounds on conditional probabilities $z_u$)\label{lem:propertiesWW}
There exist absolute constants $c\in (0, 1)$  and $C>1/c>1$ and $\epsilon_0\in (0, 1)$ such that for every $\epsilon\in (0, \epsilon_0)$ the following holds: for every vertex $v\in V$, if $y_u$ is the dual variable of a neighbor $u\in N_v(v)$ before $v$'s arrival and $\theta$ is the value chosen by \Cref{algo:TOBVC} on $v$'s arrival, then for $z_u$ as defined in~\Cref{algo:wwrounding}, we have:
	
	\begin{description}
		\item[(1)] If $\theta\not\in (c, 1-c)$, then $\sum_{u\in N_v(v)} z_u\leq 1$,
		\item[(2)] If $\theta\in [0, 1]$, then $\sum_{u\in N_v(v)} z_u           \leq 1+C\e$,
		\item[(3)] If $\sum_{u\in N_v(v)} z_u>1$, then 
%		$\sum_{u\in N_v(v)} z_u^2\leq C\e$. In particular, 
		$z_u\leq C\sqrt{\e}$ for every $u\in N_v(v)$,
		\item[(4)] If $\sum_{u\in N_v(v)} z_u>1$, then for every $u\in N_v(v)$ such that $z_u>0$, one has $y_u\in [c/2, 1-c/2]$, and
		\item[(5)] For all $u \in N_v(v)$, one has $z_u \leq 1/2 + O(\sqrt{\e})$.
	\end{description}
\end{restatable}

The following corollary will be critical to our analysis:
\begin{restatable}{cor}{probfree}\label{cor:ww}
There exist absolute constants $c>0$ and $\epsilon_0>0$ such that for all $\epsilon\in (0,\eps)$, on arrival of any vertex $v\in V$, if $z$ as defined in~\Cref{algo:wwrounding} satisfies $\sum_{u\in N_v(v)} z_u > 1$, then for every $u\in N_v(v)$ we have 
$$
c\leq \Pr[\mbox{$u$ is free when $v$ arrives}] \leq 1-c.
$$
\end{restatable}

\begin{proof}
	By~\Cref{lem:propertiesWW}, {\bf (1)} and {\bf (4)} we have that if $\sum_{u\in N_v(v)} z_u> 1$, then $\theta\in (c, 1-c)$ ($c$ is the constant from~\Cref{lem:propertiesWW}), and for every $u\in N_v(v)$ one has 
	\begin{equation}\label{eq:924hg9gfd}
	y_u\in [c/2, 1-c/2].
	\end{equation} 
	On the other hand, by~\Cref{refined-ww-xu-bound} one has 
	\begin{equation}\label{eq:23g24g43g43g}
	\frac{y_u}{\beta} \leq	x_u  \leq \frac{y_u + f(1-y_u)}{\beta},
	\end{equation}
	where $x_u$ is the fractional degree of $u$ when $v$ arrives.
	
	We now note that by~\Cref{lem:propertiesWW}, {\bf (2)}, we have that \Cref{algo:wwrounding} matches every vertex $u$ with probability at least $x_u/(1+C\e)$ (due to choices of primary arcs), and thus
	\begin{equation*}
	\begin{split}
	\Pr[\mbox{$u$ is free when $v$ arrives}]&\leq 1-\frac{x_u}{1+C\e}\\
	&\leq 1-\frac{y_u}{\beta(1+C\e)}\text{~~~~~~~(by~\eqref{eq:23g24g43g43g})}\\
	&\leq 1-\frac{c/2}{2(1+C\e)}\text{~~~~~~~(by~\eqref{eq:924hg9gfd} and the setting $\beta=2-\e\leq 2$)}\\
	&\leq 1 - c/5,
	\end{split}
	\end{equation*}
	as long as $\e$ is sufficiently small.
	
	For the other bound we will use two facts. The first is that the since $f(y)$ is monotone decreasing by \Cref{f-non-increasing} and since we picked $\beta>\beta^*(f) = 1+f(0)$, we have that for any $y\leq 1-c/2\leq 1$, 
	\begin{equation}\label{eq:5r435trdgtr4w}
	y+f(1-y)\leq 1-c/2+f(0) < \beta - c/2.
	\end{equation}
	Then, using the fact that by \Cref{line:second-pick-drop}, \Cref{algo:wwrounding} matches every vertex $u$ with probability at most $x_u$, we obtain the second bound, as follows.
	\begin{align*}
	\Pr[\mbox{$u$ is free when $v$ arrives}]&\geq 1-x_u && \\
	&\geq 1-\frac{y_u + f(1-y_u)}{\beta} && \mbox{(by~\eqref{eq:23g24g43g43g})}\\
	&\geq 1-\frac{\beta-c/2}{\beta} && \mbox{(by~\eqref{eq:924hg9gfd} and \eqref{eq:5r435trdgtr4w})}\\
	& \geq c/5. && \mbox{($\beta=2-\e<2.5$)}
	\end{align*}
	
	Choosing $c/5$ as the constant in the statement of the lemma, we obtain the result.
\end{proof}

Finally, for our analysis we will rely on the competitive ratio 
of the fractional solution maintained in \Cref{line:fractional-to-round} being $1/\beta$.
This follows by \Cref{ww-approx} and the fact that for our choices of $\beta=2-\epsilon$ and $f=f_{1+2\epsilon}$ we have that $\beta\geq \beta^*(f)$. See \Cref{app:WW-functions} for a proof of this fact.

\begin{restatable}{fact}{betafine}\label{beta-fine}
	For all sufficiently small $\epsilon>0$, we have that $2-\epsilon \geq \beta^*(f_{1+2\epsilon})$.
\end{restatable}

\subsubsection{Structural Properties of \texorpdfstring{$G_\tau$}{G-tau} and \texorpdfstring{$H_\tau$}{H-tau}}
\label{sec:structural}

In our analysis later, we focus on maximal primary paths (directed paths made of primary arcs) in $H_\tau$, 
in the sense that the last vertex along the primary path has no outgoing primary arc in $H_\tau$. 
The following definition captures termination certificates of such primary paths.
\begin{Def}[Certified Primary Path]
	A tuple $(P, T)$ is a certified primary path in $H_\tau$ if $P$ is a directed path of primary arcs in $H_\tau$ and either
	\begin{enumerate}
	  \item[{\bf (a)}] the last vertex of $P$ does not have an outgoing primary arc in $G_\tau$ and $T=\emptyset$, or 
	  \item[{\bf (b)}] the last vertex $u$ of $P$ has an outgoing primary arc $(u, w)$ in $G_\tau$ and $T=(u', w)$ is a primary arc in $H_\tau$ such that $u'$ precedes $u$ in the arrival order. 
	\end{enumerate}
  \label{def:primarypath}
\end{Def}
To elaborate, a certified primary path $(P,T)$ is made of a (directed) path $P$ of primary arcs in $H_\tau$ and $T$ is a certificate of $P$'s termination in $H_\tau$ that ensures the last vertex $u$ in $P$ has no outgoing primary arc in $H_\tau$, either due to $u$ not picking a primary arc with $T=\emptyset$, or due to the picked primary arc $(u,w)$ being blocked by another primary arc $T=(u',w)$ which appears in $H_\tau$.

As described, $G_\tau$ and $H_\tau$ differ in arcs $(u,w)$ that are \emph{blocked} by previous primary arcs to their target vertex $w$. We generally define sets of arcs which can block an edge, or a path, or a certified path from appearing in $H_\tau$ as in the following definition:
\begin{Def}[Blocking sets]  \label{def:blocking}
	For an arc $(u,w)$, define its \emph{blocking set}
	\begin{align*}
	B(u,w) := \{(u', w) \mid \mbox{$\{u',w\}$ is an edge and $u'$ arrived before $u$}\}
	\end{align*}
	to be those arcs, the appearance of any of which as primary arc in $G_\tau$ blocks $(u,v)$ from being in $H_\tau$. In other words, an arc $(u,v)$ is in $H_\tau$ as primary or secondary arc if and only if $(u,v)$ is in $G_\tau$ and none of the arcs in its blocking set $B(u,v)$ is in $G_\tau$ as a primary arc. 
	
	The blocking set of a path $P$ is simply the union of its arcs' blocking sets, 
	\begin{align*}
	B(P) := \bigcup_{(u,v) \in P} B(u,v)\,.
	\end{align*}	

	The blocking set of a certified primary path $(P,T)$ is the union of blocking sets of $P$ and $T$,
	\begin{align*}
	B(P, T) := B(P \cup T).
	\end{align*}
\end{Def}
The probability of an edge, or path, or certified primary path appearing in $H_\tau$ is governed in part by the probability of arcs in their blocking sets appearing as primary arcs in $G_\tau$. As an arc $(v,u)$ is picked as primary arc by when $v$ arrives with probability roughly $z_u$ (more precisely, $z'_u \in[z_{vu}/(1+C\epsilon), z_{vu}]$, by \Cref{lem:propertiesWW}), it will be convenient to denote by $z(v,u)$ and $z'(v,u)$ the values $z_u$ and $z'_u$ when $v$ arrives, and by $z(S)=\sum_{s\in S} z(s)$ and $z'(S)=\sum_{s\in S} z(s)$ the sum of $z$- and $z'$-values of arcs in a set of arcs $S$.

\textbf{Product distributions.} 
Note that by definition the distribution over primary and secondary arc choices of vertices are product distributions (they are independent). As such, their joint distribution is defined by their marginals.
Let $p_w$ and $s_w$ denote the distribution on primary and secondary arc choices of $w$, respectively. That is, for every $u\in N_w(w)$, $p_w(u)$ is the marginal probability that $w$ 
selects $(w, u)$ as its primary arc, and $s_w(u)$ is the marginal probability that $w$ selects $(w,u)$ as its secondary arc. 
Given our target bound \eqref{eq:enough}, it would be useful to show that conditioning on $F_u$ preserves the independence of these arc choices. Unfortunately, conditioning on $F_u$ does not preserve this independence. We will therefore refine our conditioning later on the existence of primary paths in $H_\tau$, which as we show below maintains independence of the arc choices.

\begin{lem}\label{lm:conditional-product}
For a certified primary path $(P,T)$ let $\EventEQ{(P,T)}$ be the event that the path $P$ equals a maximal connected component in $H_\tau$ and the termination of $P$ is certified by $T$.
Then the conditional distributions of primary and secondary choices conditioned on $\EventEQ{(P,T)}$ are product distributions; i.e., these conditional choices are independent. Moreover, if we let $\tilde{p}_w$ and $\tilde{s}_w$ denote the conditional distribution on primary and secondary choices of $w$, respectively, then 
    \begin{align*}
      \TV(p_w, \tilde{p}_w) \leq z(R(w)) \qquad \mbox{and} \qquad \TV(s_w, \tilde{s}_w) \leq z(R(w)),
    \end{align*}
	where $R(w) \subseteq \{w\}\times N_w(w)$ is the set of arcs leaving $w$ whose existence as primary arcs in $G_\tau$ is ruled out by conditioning on $\EventEQ{(P,T)}$, and the union of these $R(w)$, denoted by $R(P,T)$, satisfies
	\begin{equation}\label{eq:ruled-out-sets-ub}
	R(P,T):= \bigcup_w R(w) \subseteq B(P,T) \cup \{(w,r) \mid r \mbox{ is root of } P\} \cup \bigcup_{w\in P\cup\{w:\, T=(w,w')\}} \{w\}\times N_w(w).
	\end{equation}
\end{lem}

\begin{proof}
	We first bound the total variation distance between the conditional and unconditional distributions. For primary choices, conditioning on $\EventEQ{(P,T)}$ rules out the following sets of primary arc choices.
	For vertex $w\notin P$ arriving before the root $r$ of $P$ this conditioning rules out $w$ picking any edge in $B(P,T)$ as primary arc.
	For vertices $w\notin P$ with $w$ arriving after the root $r$ of $P$ this conditioning rules out picking arcs $(w,r)$. 
	Finally, this conditioning rules out 
	some subset of arcs leaving vertices in $P\cup \{w:\, T=(w,w')\}$.
	Taking the union over these supersets of $R(w)$, we obtain \eqref{eq:ruled-out-sets-ub}.
	Now, the probability of each ruled out primary choice $(w,u)\in R(w)$ is  zero under $\tilde{p}_w$ and $z'(w,u)$ under $p_w$, and all other primary choices have their probability increase, with a total increase of $\sum_{(w,u)\in R(w)} z'(w,u)$, from which we conclude that $$\TV(p_w,\tilde{p}_w) = \frac{1}{2}\sum_{u\in N_w(w)} |p_w(u) - \tilde{p}_w(u)| = z'(R(w)) \leq z(R(w)).$$
	The proof for secondary arcs is nearly identical, the only differences being that the sets of ruled out secondary arcs can be smaller (specifically, secondary arcs to $w'$ such that $T=(u,w')$ are not ruled out by this conditioning), and the probability of any arc $(w,u)$ being picked as secondary arc of $w$ is at most $\sqrt{\epsilon}\cdot z'(w,u) \leq z(w,u)$.
	
	Finally, we note that primary and secondary choices for different vertices are independent. Therefore, conditioning on each vertex $w$ not picking a primary arc in its ruled out set $R(w)$ still yields a product distribution, and similarly for the distributions over secondary choices.
\end{proof}

It is easy to show that a particular certified primary path $(P,T)$ with high value of $z(B(P,T))$ is unlikely to appear in $H_\tau$, due to the high likelihood of arcs in its breaking set being picked as primary arcs. The following lemma asserts that the probability of a vertex $u$ being the root of \emph{any} primary certified path $(P,T)$ with high $z(B(P,T))$ value is low.
\begin{lem}\label{lem:fat}
	For any $k\geq 0$ and any vertex $u$, we have the following
	\[
	\Pr[\mbox{$H_\tau$ contains any certified primary path $(P, T)$ with $P$ rooted at $u$ and $z(B(P, T)) \geq k$}] \leq e^{-k/2},
	\]
	\[
	\Pr[\mbox{$H_\tau$ contains any primary path $P$ rooted at $u$ with $z(B(P)) \geq k$}] \leq e^{-k/2}.
	\]
\end{lem}
\begin{proof}
  We first prove the bound for certified primary paths.
  For a certified primary path $(P,T)$ where the last vertex of $P$ is $w$, define $P^\ast$ as follows:
  	\begin{align*}
	P^\ast = 	\begin{cases} P & \text{if $T = \emptyset$} \\
	  \text{$P \cup \{(w,w'')\}$} & \text{if $T = (w',w'')$}.
	\end{cases}
\end{align*}

Observe that $z(B(P^\ast)) \geq k$ whenever $z(B(P,T)) \geq k$. This is trivial when $T = \emptyset$.	
To see this for the case
 $T = (w',w'')$, let $w$ be the last vertex of $P$, and note that $B(w', w'') \subseteq B(w, w'')$, 
as $w$ arrives after $w'$.
Also note that for $(P,T)$ to be in $H_\tau$, we have that $P^\ast$ must be in $G_\tau$.
 
We say a directed primary path $P' = u \to u_1 \to \cdots \to u_{\ell - 1} \to u_\ell$ is \emph{$k$-minimal} if $z(B(P')) \geq k$ 
and $z(B(P' \setminus \{(u_{\ell - 1}, u_\ell) \})) < k$. 
For such a path $P'$, define $B^\ast(P')$ as follows: Initially set $B^\ast(P') = B(P \setminus \{(u_{\ell - 1}, u_\ell )\})$.
Then from $B(u_{\ell - 1}, u_\ell)$, the breaking set of the last arc of $P'$, add arcs to $B^\ast(P')$ in reverse order of their sources' arrival until $z(B^\ast(P')) \geq k$.

Consider a certified primary path $(P,T)$ with $P$ rooted at $u$. If a $k$-minimal path rooted at $u$ which is not a prefix of $P^\ast$ is contained in $G_\tau$, then $(P,T)$ does not appear in $G_\tau$, and therefore it does not appear in $H_\tau$. On the other hand, 
if $z(B(P,T))\geq k$ then for $(P,T)$ to appear in $H_\tau$, we must have that the (unique) $k$-minimal prefix $P'$ of $P^\ast$ must appear in $G_\tau$, and that none of the edges of $B^\ast(P')$ appear in $G_\tau$. Moreover, for any certified primary path with $z(B(P,T))$, conditioning on the existence of $P'$ in $G_\tau$ does not affect random choices of vertices with outgoing arcs in $B^\ast(P')$, as these vertices are not in $P'$. Since by \Cref{lem:propertiesWW} each arc $(w,w')$ appears in $G_\tau$ with probability $z'(v,u) \geq z(v,u)/(1+C\epsilon)\geq z(v,u)/2$, we conclude that for any $k$-minimal primary path $P'$ rooted at $u$, we have
	\begin{align*}
	  &\Pr[H_\tau \text{ contains any certified primary path $(P,T)$ with $z(B(P,T)) \geq k$} \mid \text{$P'$ is in $G_\tau$}  ]\\ 
	\leq &\Pr[\text{No edge in $B^\ast(P')$ is in $G_\tau$} \mid  \text{$P'$ is in $G_\tau$}] \\
	= &\prod_{w \notin P'} (1 - \Pr[\text{Some primary edge in $B^\ast(P') \cap (\{w\} \times N_w(w))$ \text{ is in } $G_\tau$}]) \\
	\leq &\prod_{w \notin P'} \exp\left(\text{$-\textstyle\sum_{(w,w') \in B(P,T) \times N_w(w)} z(w,w')/2$}\right) \\
	\leq &\exp(-z(B^\ast(P'))/2) \leq e^{-k/2}.
	\end{align*}

	Taking total probability $\mathcal{P}_u$, the set of all $k$-minimal primary paths $P'$ rooted at $u$, we get that indeed, since $u$ is the root of at most one $k$-minimal primary path in any realization of $G_\tau$, 
	\begin{align*}
	  & \Pr[H_\tau \mbox{ contains a certified primary path $(P, T)$ rooted at $u$ with $z(B(P, T)) \geq k$}] \\
	 \leq & \sum_{P' \in \mathcal{P}_u} \underbrace{\Pr[\text{$H_\tau$ contains a $(P,T)$ with $z(B(P,T))\geq k$} \mid \text{$P'$ is in $G_\tau$}]}_{\leq \,\, e^{-k/2}} \cdot \Pr[\text{$P'$ is in $G_\tau$}]
	\leq e^{-k/2}.
	\end{align*}

	The proof for primary path is essentially the same as the above, taking $P^*=P$.
\end{proof}

\subsubsection{Analyzing Good Vertices}
\label{sec:boundvar}
Consider the set of vertices that are unlikely to be roots of long directed paths of primary arcs in $H_\tau$. 
In this section, we show that~\Cref{algo:wwrounding} achieves almost lossless rounding for such vertices, and hence we
call them \emph{good} vertices.
We start with a formal definition:
\begin{Def}[Good vertices]
We say that a vertex $v$ is \emph{good} if
\begin{align*}
\Pr_{\tau}[\mbox{$H_\tau$ has a primary path rooted at $v$ of length at least $2000 \cdot \ln(1/\varepsilon)$}]& \leq \varepsilon^6.
\end{align*}
Otherwise, we say $v$ is \emph{bad}.
\end{Def}

As the main result of this section, for good vertices, we prove the following:
\begin{thm}
	\label{thm:good_vertices}
	Let $v$ be a good vertex. Then 
	\begin{align*}
	\Pr[\mbox{$v$ is matched on arrival}] \geq (1-\eps^2) \cdot \sum_{u\in N_v(v)} x_{uv}.
	\end{align*}
\end{thm}
\paragraph{Notational conventions.} Throughout this section, we fix $v$ and let $z, z'$ be as in \Cref{algo:wwrounding}. Moreover, for simplicity of notation, we suppose that the stream of vertices ends just before $v$'s arrival and so quantities, such as  $G_\tau$ and  $H_\tau$, refer to their values when $v$ arrives. For a vertex $u$, we let $F_u$ denote the event that $u$ is free (i.e., unmatched) when $v$ arrives. In other words, $F_u$ is the event that $u$ is free in the stream that ends just before $v$'s arrival. 

To prove the theorem, first note that it is immediate if $\sum_{u\in N_v(v)} z_u \leq 1$: in that case, we have $z' = z$ and so the probability to match $v$ by a primary edge, by definition of $z_u$, is simply
\begin{align*}
\sum_{u\in N_v(v)} z_u \cdot \Pr[F_u] = \sum_{u\in N_v(v)} x_{uv}.
\end{align*}

From now on we therefore assume $\sum_{u\in N_v(v)} z_u >1$, which implies 
\begin{enumerate}[(I)]
\item $\sum_{u\in N_v(v)} z'_u = 1$,\label{en:d0}
\end{enumerate} 
and moreover, by~\Cref{lem:propertiesWW} and~\Cref{cor:ww}, for every $u\in N_v(v)$:
\begin{enumerate}[(I)]
    \addtocounter{enumi}{1}
	\item $z_u \leq C \sqrt{\eps}$, \label{en:d1}
	\item $z_u \leq (1+C\eps) \cdot z'_u$, and \label{en:d2}
	\item $c \leq \Pr[F_u] \leq 1- c$\,, \label{en:d3}
\end{enumerate}
where $c$ is the constant of \Cref{cor:ww} and $C$ is the constant of \Cref{lem:propertiesWW}. 

We now state the key technical lemma in the proof of~\cref{thm:good_vertices}: 
\begin{lem}	\label{lem:variance}
	Consider a neighbor $u\in N_v(v)$ such that 
	\begin{align}
	\Pr_{\tau}[\mbox{$H_\tau$  has a primary path rooted at $u$ of length at least $2000 \cdot \ln(1/\varepsilon)$} \mid F_u]\leq \eps^2\,.
	\label{eq:key_lemma}
	\end{align}
	Then,
	\begin{align}
	\sum_{w\in N_v(v)} z'_w \cdot \Pr[F_w \mid F_u] - \sum_{w\in N_v(v)} z'_w  \cdot \Pr[F_w]    \leq \eps^{1/3}\,.
	\label{eq:variance}
	\end{align}
\end{lem}
Note that the above lemma bounds the quantity $\sum_{w\in N_v(v)} z'_w \cdot \Pr[F_w \mid F_u]$, which will allow us to show that~\eqref{eq:enough} holds and thus the edge $\{u,v\}$ is picked in the matching with probability very close to $x_{uv}$. Before giving the proof of the lemma, we give the formal argument why the lemma implies the theorem.

\begin{proof}[Proof of~\cref{thm:good_vertices}]
	Define $S$ to be the neighbors $u$ in $N_{v}(v)$ satisfying
	\begin{align*}
	\Pr_{\tau}[\mbox{$H_\tau$ has a primary path rooted at $u$ of length at least $2000 \cdot \ln(1/\varepsilon)$} \mid F_u] > \eps^2\,.
	\end{align*}
	In other words, $S$ is the set of neighbors of $v$ that violate~\eqref{eq:key_lemma}.
	As $v$ is good, we have 
	\begin{align*}
	\eps^6 &\geq \Pr_{\tau}[\mbox{$H_\tau$ has a primary path rooted at $v$ of length at least $2000 \cdot \ln(1/\varepsilon)$}] \\
	&  \geq \sum_{u\in N_v(v)} z_u' \cdot \Pr[F_u] \cdot \Pr_{\tau} [\mbox{$H_\tau$ has a primary path rooted at $u$ of length at least $2000 \cdot \ln(1/\varepsilon)-1$} \mid F_u] \\ 
	&  \geq \sum_{u\in N_v(v)} z_u' \cdot \Pr[F_u] \cdot \Pr_{\tau} [\mbox{$H_\tau$ has a primary path rooted at $u$ of length at least $2000 \cdot  \ln(1/\varepsilon)$} \mid F_u] \\
	& \geq \sum_{u\in S} z_u' \cdot \Pr[F_u] \cdot \eps^2.
	\end{align*}
	The second inequality holds because $v$ selects the primary arc $(u,v)$ with probability $z'_u$ and, conditioned on $F_u$, $u$ cannot already have an incoming primary arc, which implies that $(u,v)$ is present in $H_\tau$. The last inequality follows from the choice of $S$.
	
	By Property~\eqref{en:d2}, $z_u \leq (1+C\eps) \cdot z'_u$ and so by rewriting we get
	\begin{align*}
	\sum_{u\in S} x_{uv} = \sum_{u\in S} z_u \cdot \Pr[F_u] \leq  (1+C\eps) \cdot \sum_{u\in S} z'_u \cdot \Pr[F_u] \leq (1+C\eps) \cdot \eps^4 \leq \eps^3.
	\end{align*}
	
	In other words, the contribution of the neighbors of $v$ in $S$ to $\sum_{u\in N_v(v)} x_{uv}$ is insignificant compared to the contribution of all neighbors, 
	\begin{align}
	\sum_{u\in N_v(v)} x_{uv}  = \sum_{u\in N_v(v)} z_u  \cdot \Pr[F_u] \geq c, 
	\label{eq:total_large}
	\end{align}
	where the inequality follows by the assumption $\sum_{u\in N_v(v)} z_u \geq 1$ and $\Pr[F_u] \geq c$ by Property~\eqref{en:d3}.
	
	We proceed to analyze a neighbor $u\in N_v(v) \setminus S$.  Recall that it is enough to verify~\eqref{eq:enough} to  conclude that edge $\{u,v\}$ is picked in the matching with probability $x_{uv}$. We have that
	\begin{align*}
	%	& \geq 1 + \sqrt{\varepsilon} \sum_{w\in N_v(v)} z'_w (1-\Pr[F_w|F_v]) - c\eps  & \mbox{($z'_u \leq z_u \leq c\sqrt{\eps}$ by~\eqref{en:d1}) }\\
	& ~~~~ 1 + \sqrt{\varepsilon} \sum_{w\in N_v(v)} z'_w \cdot (1-\Pr[F_w \mid F_u]) \\
	& \geq 1 + \sqrt{\varepsilon} \sum_{w\in N_v(v)} z'_w \cdot (1-\Pr[F_w]) -  \sqrt{\eps}\cdot\eps^{1/3} & \mbox{(by~\cref{lem:variance})}\\
	& \geq 1 + \sqrt{\varepsilon} \sum_{w\in N_v(v)} z'_w \cdot c -  \sqrt{\eps}\cdot\eps^{1/3} & \mbox{($\Pr[F_w] \leq 1- c$ by~\eqref{en:d3})}\\	
	& = 1 + \sqrt{\varepsilon} c - \sqrt{\eps}\cdot\eps^{1/3}  & \left(\sum_{w\in N_v(v)} z'_w = 1 \text{ by \eqref{en:d0}}\right)\\
	& \geq 1 + C \eps  & \left(\mbox{for $\eps$ small enough}\right)\\
	& \geq z_u/z'_u. & \left(\mbox{by \eqref{en:d2}}\right)
	\end{align*}
	Therefore, by definition of $S$ and \Cref{lem:variance}, we thus have that for every $u\in N_{v}(v)\setminus S$, the edge $\{u,v\}$ is taken in the matching with probability $x_{uv}$. Thus, the probability that $v$ is matched on arrival is, as claimed, at least
	\begin{align*}
	\sum_{u \in N_v(v) \setminus S} x_{uv}
	= \sum_{u\in N_v(v)} x_{uv} - \sum_{u\in S} x_{uv}
	\geq \sum_{u\in N_v(v)} x_{uv} - \eps^3  
	\geq (1-\eps^2) \sum_{u\in N_v(v)} x_{uv}\,,
	\end{align*}
	where the last inequality holds because  we have $\sum_{u\in N_v(v)} x_{uv} \geq c$, as calculated in~\eqref{eq:total_large}. 
\end{proof}

\subsubsection{Proof of the Key Lemma}
It remains to prove the key lemma, \Cref{lem:variance}, which we do here.

\begin{proof}[Proof of \Cref{lem:variance}] For a certified primary path $(P,T)$ let $\EventEQ{(P,T)}$ be the event as defined in~\cref{lm:conditional-product}, and let $\EventIN{(P,T)}$ be the event that $P$ is a maximal \emph{primary} path in $H_\tau$ and the termination of $P$ is certified by $T$. Further, let 
	\begin{align*}
	\mathcal{C} = \{(P,T) : \mbox{$(P,T)$ is a certified primary path rooted at $u$ with $\Pr[ \EventIN{(P,T)} ] > 0$}\}
	\end{align*}
	be the set of certified primary paths rooted at $u$  that have a nonzero probability of being maximal in $H_\tau$.  Then, by the law of total probability and since $\sum_{(P,T)\in \mathcal{C}} \Pr[\EventIN{(P,T)} \mid F_u] = 1$ 
	(since conditioning on $F_u$ implies in particular that $u$ has no incoming primary arc),
	we can rewrite the expression to bound, $\sum_{w\in N_v(v)} z'_w \cdot \Pr[F_w \mid F_u] - \sum_{w\in N_v(v)} z'_w \cdot \Pr[F_w]$, as
	\begin{align}
	\label{eq:large_variance_expression}
	\sum_{(P,T)\in \mathcal{C}}\Pr[\EventIN{(P,T)} \mid F_u] \left(\sum_{w\in N_v(v)} z'_w \cdot \Pr[F_w \mid F_u, \EventIN{(P,T)}] - \sum_{w\in N_v(v)} z'_w \cdot \Pr[F_w] \right).
	\end{align}
	We analyze this expression in two steps. First, in the next claim, we show that we can focus on the case when the certified path $(P,T)$ is very structured and equals the component of $u$ in $H_\tau$. We then analyze the sum in that structured case.
	
	\begin{claim}\label{claim:focus-on-EQ} Let  
		$\mathcal{P} \subseteq \mathcal{C}$ contain those certified primary paths $(P,T)$ of $\mathcal{C}$ that satisfy:  $P$ has length less than $2000 \cdot  \ln(1/\eps)$ and  $z(B(P,T)) \leq 2\ln(1/\eps)$. Then, we have
		\begin{align*}
		\eqref{eq:large_variance_expression} \leq 
		\sum_{(P,T)\in \mathcal{P}}\Pr[\EventEQ{(P,T)} \mid F_u] \left(\sum_{w\in N_v(v)} z'_w \cdot \Pr[F_w \mid \EventEQ{(P,T)}] - \sum_{w\in N_v(v)} z'_w \cdot \Pr[F_w] \right) + \eps^{1/3}/2.
		\end{align*}
	\end{claim}
	\begin{proof}
		Define the following subsets of certified primary paths rooted at  $u$:
		\begin{align*}
		\mathcal{C}_1 &= \{(P,T)\in \mathcal{C} \mid \mbox{$P$ is of length at least $2000 \cdot \ln(1/\eps)$}\} \\
		\mathcal{C}_2 &= \{(P,T)\in \mathcal{C} \setminus \mathcal{C}_1 \mid \mbox{$z(B(P,T)) > 2\ln(1/\eps)$}\}
		\end{align*}
		Note that $\mathcal{P} = \mathcal{C} \setminus \left( \mathcal{C}_1 \cup \mathcal{C}_2 \right)$.
		Since $u$ satisfies~\eqref{eq:key_lemma}, we have that
		\begin{align*}
		\sum_{(P,T)\in \mathcal{C}_1} \Pr[\EventIN{(P,T)} \mid F_u] \leq \eps^2 \leq \eps^{1/3}/6.
		\end{align*}
		On the other hand, by~\cref{lem:fat} and $\Pr[F_u] \geq c$ (by Property~\eqref{en:d3}), we have that
		\begin{align*}
		\sum_{(P,T)\in \mathcal{C}_2} \Pr[\EventIN{(P,T)} \mid F_u] \leq c^{-1}\cdot \sum_{(P,T)\in \mathcal{C}_2} \Pr[\EventIN{(P,T)}] \leq c^{-1} \cdot \eps \leq \eps^{1/3}/6.
		\end{align*}
		In other words, almost all probability mass lies in those outcomes where one of the certified paths $(P,T)\in \mathcal{P}$ is in $H_\tau$. It remains to prove that, in those cases, we almost always have that the component of $u$ in $H_\tau$ equals the path $P$ (whose termination is certified by $T$). Specifically, let $\overline{\EventEQ{(P,T)}}$ denote the complement of $\EventEQ{(P,T)}$. We show
		\begin{align}
		\label{eq:notequal}
		\Pr\left[\overline{\EventEQ{(P,T)}}\mid \EventIN{(P,T)}\right] \leq \eps^{1/3}/7\,.
		\end{align}
		To see this, note that by the definition of the event $\EventIN{(P,T)}$, if we restrict ourselves to primary edges then the component of $u$ in $H_\tau$ equals $P$. We thus have that for the event $\overline{\EventEQ{(P,T)}}$ to be true at least one of the vertices in $P$ must have an incoming or outgoing \emph{secondary} edge. Hence  the expression
		$\Pr\left[\overline{\EventEQ{(P,T)}}\mid \EventIN{(P,T)}\right]$ can be upper bounded by
		\begin{align}
		\Pr[\mbox{a vertex in $P$ has an incoming or outgoing secondary arc in $G_\tau$} \mid \EventIN{(P,T)}]
		\label{eq:secondary}
		\end{align}
		Note that event $\EventIN{(P,T)}$ is determined solely by choices of primary arcs. By independence of these choices and choices of secondary arcs, conditioning on $\EventIN{(P,T)}$ does not affect the distribution of secondary arcs. So the probability that any of the nodes in $P$ selects a secondary edge is at most $\sqrt{\epsilon}$. Thus, by union bound, the probability that any of the $|P|\leq 2000 \cdot \ln(1/\epsilon)$ vertices in $P$ pick a secondary arc is at most $\sqrt{\epsilon}\cdot 2000 \cdot \ln(1/\epsilon)$. We now turn our attention to incoming secondary arcs. First, considering the secondary arcs that go into $u$, we have 
		\begin{align*}
		c\leq \Pr[F_u] \leq  \prod_{(w,u)\in B(v,u)} \left( 1- z(w,u)/2 \right) \leq \exp(-z(B(v,u))/2),
		\end{align*}
		because any arc $(w,u)\in B(v,u)$ appears as a primary arc in $G_\tau$ independently with probability at least $z(w,u)/2$ and the appearance of such an arc implies that $u$ has an incoming primary arc in $H_\tau$ and is therefore matched; i.e., the event $F_u$ is false in this case. 
		We thus have $z(B(v,u)) \leq 2\ln(1/c)$.  Further, since $(P,T) \not \in \mathcal{C}_2$, we have $z(B(P)) \leq z(B(P,T)) \leq 2\ln(1/\eps)$. Again using that the conditioning on $\EventIN{(P,T)}$ does not affect the distribution of secondary edges, we have that the probability of an incoming secondary arc to any vertex in $P$ is at most $\sqrt{\eps} \cdot \left( 2\ln(1/c) + 2\ln(1/\eps) \right).$
		Thus, by union bound, the probability that any vertex in $P$ has an incoming or outgoing secondary arc conditioned on $\EventIN{(P,T)}$ is at most 
		\begin{align*}
		\sqrt{\epsilon}\cdot 2000 \cdot \ln(1/\epsilon) +  \sqrt{\eps} \cdot \left( 2\ln(1/c) + 2\ln(1/\eps) \right) \leq \eps^{1/3}/7,
		\end{align*}
		for sufficiently small $\epsilon$, which implies~\eqref{eq:notequal} via~\eqref{eq:secondary}.
		
		We now show how  the above concludes the proof of the claim.  We have shown that each one of the two sets $\mathcal{C}_1, \mathcal{C}_2$  contributes at most $\eps^{1/3}/6$ to~\eqref{eq:large_variance_expression} (where we use that  $\sum_{w\in N_v(v)} z'_w = 1$). Hence,
		\begin{align*}
		\eqref{eq:large_variance_expression} &\leq 
		\sum_{(P,T)\in \mathcal{P}}\Pr[\EventIN{(P,T)} \mid F_u] \left(\sum_{w\in N_v(v)} z'_w \cdot \Pr[F_w \mid \EventIN{(P,T)}, F_u] - \sum_{w\in N_v(v)} z'_w \cdot \Pr[F_w] \right) + 2\eps^{1/3}/6.
		\end{align*}
		This intuitively concludes the proof of the claim as~\eqref{eq:notequal} says that $\Pr[\EventEQ{(P,T)}| \EventIN{(P,T)}]$ is almost $1$. The formal calculations are as follows.
		Since the event $\EventEQ{(P,T)}$ implies the event $\EventIN{(P,T)}$, we have that
		\begin{align*}
		\Pr[\EventEQ{(P,T)}] = \Pr[\EventEQ{(P,T)} \wedge \EventIN{(P,T)}]  = \Pr[\EventIN{(P,T)} ] - \Pr[\overline{\EventEQ{(P,T)}} \wedge \EventIN{(P,T)}], 
		\end{align*}
		which by~\eqref{eq:notequal} implies
		\begin{align}
		\Pr[\EventEQ{(P,T)} ] = \Pr[\EventIN{(P,T)} ]\left( 1- \Pr\left[\overline{\EventEQ{(P,T)}}\mid \EventIN{(P,T)} \right]\right) \geq \Pr[\EventIN{(P,T)} ] \left(1- \eps^{1/3}/7   \right).\label{eq:EQ-in-terms-of-IN}
		\end{align}
		We use this to rewrite $\Pr[\EventIN{(P,T)} \mid F_u] \left(\sum_{w\in N_v(v)} z'_w \cdot \Pr[F_w \mid \EventIN{(P,T)}, F_u] - \sum_{w\in N_v(v)} z'_w \cdot \Pr[F_w] \right)$. Specifically, by law of total probability, it can be rewritten as the sum of the expressions~\eqref{eq:first_express} and~\eqref{eq:second_express} below:
		\begin{align}
		 \Pr[\EventEQ{(P,T)} \wedge \EventIN{(P,T)} \mid F_u] & \left(\sum_{w\in N_v(v)} z'_w \cdot \Pr[F_w \mid \EventEQ{(P,T)}, \EventIN{(P,T)}, F_u] - \sum_{w\in N_v(v)} z'_w \cdot \Pr[F_w] \right) \nonumber\\
		 = \Pr[\EventEQ{(P,T)}  \mid F_u] & \left(\sum_{w\in N_v(v)} z'_w \cdot \Pr[F_w \mid \EventEQ{(P,T)}] - \sum_{w\in N_v(v)} z'_w \cdot \Pr[F_w] \right) 		\label{eq:first_express}
		\end{align}
		and
		\begin{align}\label{eq:second_express}
		\Pr[\overline{\EventEQ{(P,T)}} \wedge \EventIN{(P,T)} \mid F_u] \left(\sum_{w\in N_v(v)} z'_w \cdot \Pr[F_w \mid \overline{\EventEQ{(P,T)}}, \EventIN{(P,T)}, F_u] - \sum_{w\in N_v(v)} z'_w \cdot \Pr[F_w] \right), 
		\end{align}
		where \eqref{eq:second_express} can be upper bounded as follows: 
		\begin{align*}
		\eqref{eq:second_express} & \leq \Pr[\overline{\EventEQ{(P,T)}} \wedge \EventIN{(P,T)} \mid F_u] && \text{(by $\sum_{w\in N_v(v)} z'_w \leq 1$)}\\
		 & \leq c^{-1}\cdot \Pr[\overline{\EventEQ{(P,T)}} \wedge \EventIN{(P,T)} ] && \text{(by $c\leq \Pr[F_u]$)}\\
		 & = c^{-1}\cdot \Pr[\EventIN{(P,T)}  ] \cdot \Pr\left[\overline{\EventEQ{(P,T)}}\mid \EventIN{(P,T)}\right]\\
		 & \leq c^{-1}\cdot \frac{\Pr[\EventEQ{(P,T)} ]}{1-\eps^{1/3}/7}\cdot (\eps^{1/3}/7) && \text{(by \eqref{eq:notequal} and \eqref{eq:EQ-in-terms-of-IN})}\\
		 & \leq \Pr[\EventEQ{(P,T)}] \cdot \eps^{1/3}/6. && \text{(for $\epsilon$ small enough)}
		\end{align*}
		
		As at most one of the events $\{\EventEQ{(P,T)}\}_{(P,T) \in \mathcal{P}}$ is true in any realization of $G_\tau$, we have that $\sum_{(P,T) \in \mathcal{P}} \Pr[\overline{\EventEQ{(P,T)}} \wedge \EventIN{(P,T)} \mid F_u] \leq \sum_{(P,T) \in \mathcal{P}} \left( \Pr[\EventEQ{(P,T)}] \cdot \eps^{1/3}/6\right) \leq \eps^{1/3}/6$.
		Thus, again using that $\sum_{w\in N_v(v)} z_w \leq 1$, we have that
		\begin{align*}
		\eqref{eq:large_variance_expression} &\leq 
		\sum_{(P,T)\in \mathcal{P}}\Pr[\EventIN{(P,T)} \mid F_u] \left(\sum_{w\in N_v(v)} z'_w \cdot \Pr[F_w \mid \EventIN{(P,T)}, F_u] - \sum_{w\in N_v(v)} z'_w \cdot \Pr[F_w] \right) + 2\eps^{1/3}/6  \\
		& \leq \sum_{(P,T)\in \mathcal{P}}\Pr[\EventEQ{(P,T)} \mid F_u] \left(\sum_{w\in N_v(v)} z'_w \cdot \Pr[F_w \mid \EventEQ{(P,T)}] - \sum_{w\in N_v(v)} z'_w \cdot \Pr[F_w] \right) + 3\eps^{1/3}/6,
		\end{align*}
		as claimed.    
	\end{proof}

	The previous claim bounded the contribution of certified primary paths in $\mathcal{C}\setminus \mathcal{P}$ to \eqref{eq:large_variance_expression}. The following claim bounds the contribution of paths in $\mathcal{P}$.
	\begin{claim}\label{claim:bound-for-EQ}
		Let $\mathcal{P} \subseteq \mathcal{C}$ contain those certified primary paths $(P, T)$ of $\mathcal{C}$ that satisfy:  $P$ has length less than $2000 \cdot \ln(1/\eps)$ and  $z(B(P, T)) \leq 2\ln(1/\eps)$. Then, we have
		\begin{align*}
		\sum_{(P,T)\in \mathcal{P}}\Pr[\EventEQ{(P,T)}] \left(\sum_{w\in N_v(v)} z'_w \cdot \Pr[F_w \mid \EventEQ{(P,T)}] - \sum_{w\in N_v(v)} z'_w \cdot \Pr[F_w] \right)\leq \epsilon^{1/3}/2.
		\end{align*}
	\end{claim}
	
\begin{proof}	
We prove the claim in two steps: first we construct a chain of distributions that interpolates between the unconditional distribution of $H_\tau$ and its conditional distribution, and then bound the expected number of vertices that change their matched status along that chain. For the remainder of the proof we fix the certified primary path $(P,T)$.

\paragraph{Constructing a chain of distributions.} Let $H_\tau^{(0)}$ denote the unconditional distribution of $H_\tau$ when $v$ arrives, and let $H_\tau^{(n)}$ denote the distribution of $H_\tau$ conditioned on $\EventEQ{(P,T)}$ when $v$ arrives. Here $n=|V|$ is the number of vertices in the input graph. For every $w\in V$ let $F^{(0)}_w$ denote the indicator of $w$ being free when $v$ arrives (unconditionally) and let $F^{(n)}_w$ denote the indicator variables of $w$ being free when $v$ arrives conditioned on $\EventEQ{(P,T)}$. Note that $F^{(0)}$ is determined by $H^{(0)}_\tau$ and $F^{(n)}$ is determined by $H^{(n)}_\tau$.  For $t=0,\ldots, n$, we define distributions $H^{(t)}_\tau$ that interpolate between $H_\tau^{(0)}$ and $H_\tau^{(n+1)}$ as follows. 

As in Lemma~\ref{lm:conditional-product}, for every $w\in V$ we denote the unconditional distribution of its primary choice by $p_w$, and the unconditional distribution of its secondary choice by $s_w$. Similarly, we denote the conditional distribution given $\EventEQ{(P,T)}$ of the primary choice by $\wt{p}_w$ and the conditional distribution of the secondary choice by $\wt{s}_w$. For every $t=0,\ldots, n$ the primary choice of vertices $w_j, j=1,\ldots, t$ are sampled independently from $\wt{p}_{w_j}$, and the primary choices of vertices $w_j, j=t+1,\ldots, n$ are sampled independently from the unconditional distribution $p_{w_t}$. Similarly, secondary choices of vertices $w_j, j=1,\ldots, t$ are sampled independently from $\wt{s}_{w_j}$ and secondary choices of vertices $w_j,j=t+1,\ldots, n$ are sampled independently from $s_{w_j}$. Note that $H^{(0)}_\tau$ is sampled from the unconditional distribution of $H_\tau$, and $H_\tau^{(n)}$ is sampled from the conditional distribution (conditioned on $\EventEQ{(P,T)}$), as required, due to the independence of the conditional probabilities $\tilde{p}_{w_j}$ and $\tilde{s}_{w_j}$, by \Cref{lm:conditional-product}. For $t=0,\ldots, n$ let $M_t$ denote the matching constructed by our algorithm on $H_\tau^{(t)}$, and let $F^{(t)}_w$ be the indicator variable for $w$ being free when $v$ arrives in the DAG sampled from $H^{(t)}_\tau$.

\paragraph{Coupling the distributions of $H_\tau^{(t)}$.} We now exhibit a coupling between the $H_\tau^{(t)},t=0,\ldots, n$. Specifically, we will show that for every such $t$ the following holds.
\begin{equation}\label{eq:9243hg9h9234g}
\expect\left[\sum_{q\in V} |F^{(t+1)}_q-F^{(t)}_q|\right]\leq 4z(R(w_{t+1})),
\end{equation}
where $R(w_{t+1})$ is as defined in Lemma~\ref{lm:conditional-product} with regard to the certified primary path $R(P,T)$.
Recall that $z(R(w_{t+1}))$ is the total probability assigned to arcs leaving $w_{t+1}$ which are ruled out from being primary arcs in $G_\tau$ by conditioning on $\EventEQ{(P, T)}$.

We construct the coupling by induction. The base case corresponds to $t=0$ and is trivial. We now give the  inductive step ($t\to t+1$). We write $w:=w_{t+1}$ to simplify notation. Let $Z^p\in N_w(w)$ denote the primary choice of $w$ in $H_\tau^{(t)}$, and let $Z^s\in N_w(w)$ denote the secondary choice of $w$ in $N_w(w)$ (they are sampled according to the unconditional distributions $p_w$ and $s_w$ respectively). Let $\wt{Z}^p\in N_w(w)$ and $\wt{Z}^s\in N_w(w)$ be sampled from the conditional distributions $\wt{p}_w$ and $\wt{s}_w$ respectively, such that that the joint distributions $(Z^p, \wt{Z}^p)$ and $(Z^s, \wt{Z}^s)$ satisfy
\begin{equation}\label{eq:923y9th23t32}
\Pr[Z^p\neq \wt{Z}^p]=\TV(p_w, \wt{p}_w)\text{~~and~~}\Pr[Z^s\neq \wt{Z}^s]=\TV(s_w, \wt{s}_w).
\end{equation}

First, we note that if $Z^p=\wt{Z}^p$ and $Z^s = \wt{Z}^s$, then $w=w_{t+1}$ is matched to the same neighbor under $H_\tau^{(t)}$ and $H_\tau^{(t+1)}$, and so $M_t=M_{t+1}$, due to the greedy nature of the matching constructed. Otherwise, by Lemma~\ref{lem:freevertices}, at most two vertices have different matched status in $M_t$ and $M_{t+1}$ in the latter case (in the former case every vertex has the same matched status). To summarize, we have, for $R(w)$ determined by $(P,T)$ as in \Cref{lm:conditional-product}, that 
\begin{equation}\label{eq:239hg92hg932g}
\begin{split}
\expect\left[\sum_{q\in V} |F^{(t+1)}_q-F^{(t)}_q|\right]&\leq 2\cdot \Pr[Z^p\neq \wt{Z}^p\text{~or~}Z^s\neq \wt{Z}^s]\\
&\leq 2(\TV(p_w, \wt{p}_w)+\TV(s_w, \wt{s}_w))\text{~~~~(by~\eqref{eq:923y9th23t32} and union bound)}\\
&\leq 4z(R(w)).\text{~~~~~~~~~~~~~~~~~~~~~~~~~~~~(by~Lemma~\ref{lm:conditional-product})}
\end{split}
\end{equation}
This concludes the proof of the inductive step, and establishes~\eqref{eq:9243hg9h9234g}. In particular, we get
\begin{equation}\label{eq:239hg23g}
\begin{split}
\expect\left[\sum_{q\in V} |F^{(n)}_q-F^{(0)}_q|\right]&\leq \sum_{t=0}^{n-1} \expect\left[\sum_{q\in V} |F^{(t+1)}_q-F^{(t)}_q|\right]\\
&\leq \sum_{t=0}^{n-1} 4z(R(w_{t+1})) \text{~~~~~~~(by~\eqref{eq:239hg92hg932g})}\\
&= 4z(R(P,T)),
\end{split}
\end{equation}
by the definition of $R(P,T) = \bigcup_w R(w)$ in \Cref{lm:conditional-product}.

We now finish the claim. First note that for any $(P,T)$ such that $P$ has length at most $2000 \cdot \ln(1/\eps)$ and  $z(B(P,T)) \leq 2\ln(1/\eps)$ one has $\sum_w z(R(w)) =  z(R(P,T))=O(\ln(1/\eps))$. Indeed, by Lemma~\ref{lm:conditional-product} and linearity of $z$, recalling that $u$ is the root of $P$ and that no vertex appears after $v$ (and thus $B(v,u)=\{(w,u) \mid w \mbox{ arrives between $u$ and $v$}\}$), we have
\begin{align}\label{eq:ub-zRPT}
z(R(P, T)) & \leq z(B(P,T)) + z(B(v,u)) +  \sum_{w\in P\cup\{w:\, T=(w,w')\}} z\left( \{w\}\times N_w(w)\right).
\end{align}
We now bound the contribution to the above upper bound on $\sum_w z(R(w)) = z(R(P,T))$ in \eqref{eq:ub-zRPT}. First,
we have that $z(B(P, T))\leq 2\ln(1/\eps)$ by assumption of the lemma. 
To bound the contribution of  $z(B(v,u))$, we note that by Property~\ref{en:d3}, we have
    \begin{align*}
      c\leq \Pr[F_u] = \prod_{e\in B(v,u)} \left( 1- z_e \right) \leq \exp\left(-\sum_{(w,u)\in B(v,u)} z(w,u)/2\right) \leq \exp(-z(B(v,u))/2),
    \end{align*}
    because any arc $e = (w,u)$ appears as a primary arc in $G_\tau$ with probability $z'(w,u)\geq z(w,u)/2$, independently of other such arcs, and the appearance of any such an edge implies that $u$ has an incoming primary edge in $H_\tau$ when $v$ arrives and is therefore matched; i.e., the event $F_u$ is false in this case. We thus have $z(B(v,u)) \leq 2\ln(1/c)$. Finally, it remains to note that for every one of the at most $2000 \cdot \ln(1/\epsilon)+1$ vertices $w\in P\cup\{w:\, T=(w,w')\}$ the contribution of $z(\{w\}\times N_w(w))$ to 
    the right hand side of \eqref{eq:ub-zRPT} is at most $1+C\e\leq 2$, by \Cref{lem:propertiesWW}, \textbf{(2)}. Putting these bounds together, we get that for sufficiently small $\e$,
	\begin{equation}\label{eq:923hg92hg}
		z(R(P, T))\leq 2\ln(1/\eps)+2\ln(1/c)+2 \cdot 2000 \cdot \ln(1/\e)+2=O(\ln(1/\e)).
	\end{equation}

	The term we wish to upper bound is at most 
	\begin{align*}
		& \sum_{w\in N_v(v)} z'_w \cdot \Pr[F_w \mid \EventEQ{(P,T)}] - \sum_{w\in N_v(v)} z'_w \cdot \Pr[F_w] && \\
		\leq & \left(\max_{w\in N_v(v)} z'_w\right)\cdot \sum_{w\in N_v(v)} \left|\Pr[F_w \mid \EventEQ{(P,T)}] - \Pr[F_w]\right| && \\
		\leq & C\sqrt{\e}\cdot \sum_{w\in N_v(v)} \left|\Pr[F_w \mid \EventEQ{(P,T)}] - \Pr[F_w]\right| && \text{(by Lemma~\ref{lem:propertiesWW}, {\bf (3)})}\\
		= & C\sqrt{\e}\cdot \expect\left[ \sum_{w\in N_v(v)} |F^{(n)}_w-F^{(0)}_w|\right] && \text{(by definition of $F^{(0)}$ and $F^{(n)}$)}
	\end{align*}
	then, using \eqref{eq:239hg23g} and \eqref{eq:923hg92hg}, we find that the term we wish to upper bound is at most
	\begin{align*}
		\leq & C\sqrt{\e}\cdot \expect\left[ \sum_{w\in V}|F^{(n)}_w-F^{(0)}_w|\right] && \\
		\leq & C\sqrt{\e}\cdot z(R(P,T)) && \text{(by ~\eqref{eq:239hg23g})}\\
		= & O(\sqrt{\e}\cdot \log(1/\e)) && \text{(by ~\eqref{eq:923hg92hg})} \\
		\leq  & \epsilon^{1/3}/2,
	\end{align*}
	completing the proof.
\end{proof}
	Finally, we obtain \Cref{lem:variance} by combining \Cref{claim:focus-on-EQ} and \Cref{claim:bound-for-EQ}, to find that, as claimed
		\begin{align*}
		\eqref{eq:large_variance_expression} & \leq 
		\sum_{(P,T)\in \mathcal{P}}\Pr[\EventEQ{(P,T)} \mid F_u] \left(\sum_{w\in N_v(v)} z'_w \cdot \Pr[F_w \mid \EventEQ{(P,T)}] - \sum_{w\in N_v(v)} z'_w \cdot \Pr[F_w] \right) + \eps^{1/3}/2 \\
		& \leq \epsilon^{1/3}/2 + \epsilon^{1/3}/2 = \epsilon^{1/3}. \qedhere
	\end{align*}
\end{proof}

\subsubsection{Bounding the Impact of Bad Vertices}
\label{sec:bad-is-ok}

In this section, we show that we can completely ignore the bad vertices without losing too much. 
From the definition of good vertices, for a bad vertex $v$, we have that 
$$\Pr_{\tau}[\mbox{$H_\tau$ has a primary path rooted at $v$ of length at least 
$2000 \cdot \ln(1/\varepsilon)$}] \geq \varepsilon^6.$$ 
As the main result of this section, we prove the following theorem:
\begin{thm}
  \label{lem:bad-is-ok}
  The number of bad vertices is at most $ \varepsilon^{3} \cdot \sum_{e \in E} x_{e}$.
\end{thm}

To prove this, we first describe a charging mechanism in which, for each bad vertex, a charge of one is distributed 
among a subset of other vertices. 
Then, using the following supplementary lemma, we show that the total distributed charge over all vertices in the graph is at most $\varepsilon^{3} \cdot \sum_{(u,v) \in E} x_{uv}$.

\begin{lem}
  \label{lem:thin-long-path}
  We call a primary path $P$ a \emph{primary predecessor path} of $v$ if it ends at $v$. That is, $P = v_\ell \to v_{\ell-1} \to \dots \to v_{1} = v$. 
  We have
  \begin{align*}
	\Pr_\tau [\mbox{$v$ has any primary predecessor path $P$ with $z(B(P)) \leq 20 \cdot \ln(1/\e)$ and $|P| \geq 1000 \cdot \ln(1/\e)$}] \leq \epsilon^{10}.
  \end{align*}
\end{lem}
\begin{proof}
  We use the principle of deferred decisions and traverse the path backwards. Let $b$ be the current vertex, which is initially set to $v$. 
  Consider all incoming arcs to $b$, say $(a_1, b), \ldots, (a_k, b)$ where we index $a$'s by time of arrival; i.e., $a_i$ arrives before $a_j$ if $i<j$ (and $b$ arrived before any $a_i$).
  
  First consider the random choice of $a_1$ and  see if it selected the arc $(a_1, b)$. 
  \begin{itemize}
    \item If it does, then the path including $b$ in $H_\tau$ will use the arc $(a_1, b)$. 
    \item Otherwise, if $a_1$ does not select the arc $(a_1, b)$, then go on to consider $a_2$ and so on.
  \end{itemize}
  If no $a_1, \ldots, a_k$ selects $b$, then the process stops; i.e., the primary path starts at this vertex since $b$ has no incoming primary arc. 
  Otherwise let $i$ be the first index so that $(a_i, b)$ was selected. Then $(a_i, b)$ is in the primary path ending at $v$ in $H_\tau$. 
  Now, observe that \emph{no} $a_1, \ldots, a_{i-1}$ may be in the path in this case, because these vertices arrived before $a_i$ and after $b$. 
  Moreover, we have not revealed any randomness regarding $a_{i+1}, \ldots, a_{k}$ that may appear later in the path. 
  We can therefore repeat the above process with $b$ now set to $a_i$ and ``fresh'' randomness for all vertices we consider, as the random choices of arcs of all vertices are independent. 
  We now show that this process, with good probability, does not result in a long predecessor path $P$ of low $z(B(P))$ value.

  Recall from \Cref{lem:propertiesWW}, {\bf (5)}, that $z(u,v) \leq 3/5$ for all $(u,v) \in V\times V$. 
  Suppose that $\sum_{i = 1}^k z(a_i, b) \geq 4/5$.
  Let $j$ be the first index such that $\sum_{i = 1}^j z(a_i, b) \geq 1/5$. 
  Thus $\sum_{i=1}^{j} z(a_i, b) \leq 4/5$, and hence the probability that none of the first $j$ vertices select $b$ is at least $\prod_{i=1}^j(1 - z(a_i, b)) \geq  1 - \sum_{i=1}^j z(a_i, b) \geq 1/5$. 
  Consequently, with probability at least $1/5$, vertex $b$ either has no predecessor or the increase to $z(B(P))$ is at least $1/5$.

  In the other case, we have $\sum_{i=1}^k z(a_i, b)\leq 4/5$. Then the probability that $b$ has \emph{no} predecessor is 
  $\prod_{i=1}^k (1 - z(a_i, b)) \geq  1 - \sum_{i=1}^k z(a_i, b) \geq 1/5$.

  Therefore, at any step in the above random process, with probability at least $1/5$, we either stop or increase $z(B(P))$ by $1/5$.  
  Let $Z_i$ be an indicator variable for the random process either stopping or increasing $z(B(P))$ by at least $1/5$ at step $i$, and notice that according 
  to the above random process, each $Z_i$ is lower bounded by an independent Bernoulli variable with probability $1/5$.
  Thus if we define $Z = \sum_{i \in [1000 \cdot \ln (1/\e)]} Z_i$, we have $\expect[Z] \geq 200 \cdot \ln (1/\e)$, and thus by 
  standard coupling arguments and Chernoff bounds, we have that
    $$\Pr[Z \leq 100 \cdot \ln (1/\e)] \leq \Pr \left[Z \leq (1 - 1/2) \cdot \expect[Z] \right] \leq e^{-(1/2) \cdot (1/2)^2 \cdot 200 \cdot \ln (1/\e)} \leq \e^{10}.$$ 
 But if the path does not terminate within $1000 \cdot \ln (1/\e)$ steps and $Z \geq 100 \cdot \ln(1/\varepsilon)$, 
 then $z(B(P)) \geq 20 \cdot \ln(1/\e)$. 
\end{proof}

We now prove \Cref{lem:bad-is-ok}.

\begin{proof}[Proof of \Cref{lem:bad-is-ok}]
By \Cref{lem:fat}, the probability that $H_\tau$ has a primary path $P$ with $z(B(P)) \geq 20 \cdot \ln (1/\varepsilon)$ 
starting at $v$ is at most $\varepsilon^{10}$.
Thus, for a bad vertex $u$, the probability that $H_\tau$ has some primary  path $P$ rooted at $u$ with  $|P| \geq 2000 \cdot \ln (1 / \varepsilon)$ and $z(B(P)) \leq 20 \cdot \ln(1/\epsilon)$ is at least $\varepsilon^6 - \varepsilon^{10} \geq \varepsilon^6/2$.

Let $k = 20 \cdot \ln(1/\varepsilon)$ and $\ell = 2000 \cdot \ln(1/\varepsilon)$.
Let $\mathcal{P}_u$ be the set of all primary paths $P$ rooted at $u$ such that $z(B(P)) \leq k$ and $|P| = \ell$ starting at $u$.
Since all such primary paths with length more than $\ell$ are extensions of those with length exactly $\ell$, we have
$\sum_{P \in \mathcal{P}_u} \Pr[P \text{ is in } H_\tau]  \geq \e^6/2$. 
For each such path $P \in \mathcal{P}_u$, consider the two vertices $w^P_\ell$ and $w^P_{\ell - 1}$ at distances $\ell$ and $\ell - 1$ respectively from $u$. 
For each such vertex $w^P_j$ ($j \in \{\ell-1,\ell\}$), charge $(2/\varepsilon^6) \cdot \Pr[P \text{ is in } H_\tau] \cdot y_{w^P_j}$.
Then the sum of these charges is 
$$\sum_{P \in \mathcal{P}_u}  (2/\varepsilon^6) \cdot \Pr[P \text{ is in } H_\tau] \cdot \underbrace{(y_{w^P_\ell} + y_{w^P_{\ell-1}})}_{\geq 1} \geq (2/\e^6) \cdot \sum_{P \in \mathcal{P}_u} \Pr[P \text{ is in } H_\tau]  \geq 1.$$
Notice that the fact $(y_{w^P_\ell} + y_{w^P_{\ell-1}}) \geq 1$ follows because $y_w$'s form a feasible dual solution (to the vertex cover problem).

On the other hand, consider how many times each vertex is charged. 
For this, for every vertex $w$, let $\mathcal{Q}_w$ be the set of primary predecessor paths $Q$ of $u$ such that $|Q| = \ell - 1$ and $z(B(P)) \leq k$. 
As $|Q| = \ell - 1 \geq 1000 \cdot \ln(1/\e)$ for all $Q \in \mathcal{Q}_w$, by~\Cref{lem:thin-long-path}, $\sum_{Q \in \mathcal{Q}_w} \Pr[Q \text{ is in } H_\tau] \leq \e^{10}$ . 
For a primary predecessor path $Q \in \mathcal{Q}_w$ (or one of its extensions), the vertex $w$ can be charged at most twice according to the above charging mechanism. Since any predecessor path of $w$ with length more than $\ell - 1$ must be an extension of one with length exactly $\ell - 1$, we have that the amount $w$ is charged is at most
$$\sum_{Q \in \mathcal{Q}_w} 2 \cdot 2 \cdot \Pr[Q \text{ is in } H_\tau] \cdot y_w/ \e^6 \leq 4 \cdot (\e^{10}/\e^6) \cdot y_w \leq 4 \cdot \e^4 \cdot y_w.$$
Summing over all $w \in V$ and using \Cref{refined-ww-xu-bound}, the total charge is at most
\begin{equation*}\sum_{w \in v} 4 \cdot \e^4 \cdot y_w \leq 4 \cdot \e^4 \cdot \beta \cdot \sum_{e \in E} x_{e} \leq \e^3 \sum_{e \in E} x_{e}. \qedhere
\end{equation*}
\end{proof}

\subsubsection{Calculating the Competitive Ratio of \Cref{algo:wwrounding}}\label{sec:competitive}

We now show that the competitive ratio of~\Cref{algo:wwrounding} is indeed $(1/2 + \alpha)$ competitive for some sufficiently small absolute constant
$\alpha > 0$, thus proving~\Cref{thm:vertex-intro}. This essentially combines the facts that for good vertices, the matching probability is very close to the fractional values of incident edges, and that the number of bad vertices is very small compared to the total value of the fractional algorithm (over the entire graph).

\begin{proof}[Proof of~\Cref{thm:vertex-intro}]

Let $\operatorname{OPT}$ denote the size of the maximum cardinality matching in the input graph $G$. Then, by~\Cref{ww-approx} and our choice of $f=f_{1+2\epsilon}$ and $\beta=2-\eps \geq \beta^*(f_{1+2\epsilon})$, we have that
$\sum_{e} x_{e} \geq (1/\beta) \cdot \operatorname{OPT} \geq (1/2 + \e/4) \cdot \operatorname{OPT}$, where the $x_{e}$'s are the fractional values we compute in~\Cref{algo:wwrounding}.

Now let $M$ be the matching output by~\Cref{algo:wwrounding}. We have
  \begin{align*}
	\expect[{|M|}] &= \sum_{e \in E} \Pr[e \text{  is matched}] \\
	&\geq \sum_{\text{good } v \in V} (1 - \e^2) \cdot \sum_{u \in N_v(v)} x_{uv} & \text{(By~\Cref{thm:good_vertices})} \\
	&\geq  (1 - \e^2) \cdot \left(\sum_{e \in E} x_e - \sum_{\text{bad } v \in V} \sum_{u \in N_v(v)} x_{uv} \right)
	\\
	&\geq  (1 - \e^2) \cdot \left(\sum_{e \in E} x_e - \sum_{\text{ bad } v \in V} 1 \right) 	& \text{($\sum_{u \in N_v(v)} x_{uv} \leq 1$)} \\	
	&\geq  (1 - \e^2) \cdot \left(\sum_{e \in E} x_e - \e^3 \sum_{e \in E} x_{e} \right) 	& \text{(By~\Cref{lem:bad-is-ok})} \\
	&\geq  (1 - 2\e^2) \cdot \sum_{e \in E} x_{e} \\
	&\geq  (1 - 2\e^2) \cdot (1/2 + \e/4) \cdot \operatorname{OPT} \\
	&\geq  (1/2 + \e/5) \cdot \operatorname{OPT}, 
  \end{align*}
  where the last line holds for a sufficiently small constant $\epsilon>0$.
\end{proof}

\appendix
\section*{Appendix}
%!TEX root = ./main.tex
\section{Deferred Proofs of \Cref{sec:imrpoved-algorithm}}\label{app:imrpoved-algorithm}
Here we prove that a change of the realized arc choices of any vertex does not change the matched status of more than two vertices (at any point in time). This is \Cref{lem:freevertices}, restated below.

\freevertices*
\begin{proof}
	We consider the evolution, following each vertex arrival, of the matchings $M_\tau$ and $M_{\tau'}$ computed in $H_\tau$ and $H_{\tau'}$, respectively, as well as the set of vertices with different matched status in these matchings, denoted by $D:=(M_{\tau} \setminus M_{\tau'})\cup (M_{\tau'} \setminus M_{\tau})$. 	
	The set $D$ is empty before the first arrival and remains empty until the arrival of $v$, as all earlier vertices than $v$ have the same primary and secondary arcs and have the same set of free neighbors in $H_\tau$ and $H_{\tau'}$ (as $D=\emptyset$, by induction). Now, if immediately after $v$ arrives it remains free in both $M_\tau$ and $M_{\tau'}$, or it is matched to the same neighbor in both matchings, then clearly $D$ remains empty. Otherwise, either $v$ is matched to different neighbors in $M_\tau$ and $M_{\tau'}$, or $v$ is matched in one of these matchings but not in the other. Both these cases result in $|D|=2$. We now show by induction that the cardinality of $D$ does not increase following subsequent arrivals, implying the lemma.
	
	Let $u$ be some vertex which arrives after $v$. 
	If when $u$ arrives $u$ is matched to the same neighbor $w$ in $M_\tau$ and $M_{\tau'}$ or if $u$ remains free in both matchings, then $D$ is unchanged.
	If $u$ is matched to some $w$ on arrival in $M_\tau$, but not in $M_{\tau'}$, then since the arcs of $u$ are the same in $G_\tau$ and $G_{\tau'}$, this implies that $w$ must have been free in $M_\tau$ but not in $M_{\tau'}$, and so $D\ni w$. Therefore, after $u$ arrives, we have $D \leftarrow (D\setminus \{w\}) \cup \{u\}$, and so $D$'s cardinality is unchanged.
	Finally, if $u$ is matched to two distinct neighbors, denoted by $w$ and $w'$, respectively, then one of $(u,w)$ and $(u,w')$ must be the primary arc of $u$ in both $G_\tau$ and $G_{\tau'}$. Without loss of generality, say $(u,w)$ is this primary arc. Since $u$ is matched to $w$ in $M_\tau$ but not in $M_{\tau'}$, then $w$ must be free in $M_\tau$ when $u$ arrives, but not in $M_{\tau'}$, and so $D\ni w$. Consequently, we have that after $u$ arrives we have $D \leftarrow S$ for some set $S \subseteq (D\setminus \{w\}) \cup \{w'\}$, and so $D$'s cardinality does not increase.
\end{proof}

\section{Deferred Proofs of \texorpdfstring{\Cref{sec:fractional-algo}}{Algorithm 1}}\label{app:fractional-algo}
Here we prove the bound on the fractional degree $x_u$ in terms of its dual value, restated below.
\refinedwwbound*
\begin{proof}
	Let $y_0$ be $u$'s potential after $u$'s arrival. For the lower bound, note that it suffices to prove that every increase in the fractional degree is bounded below by the increase in the potential divided by $\beta$. When vertex $u$ first arrived, we consider two cases.
	
	\begin{enumerate}
		
		\item  $y_0  > 0$ (thus $y_0 = 1 - \theta > 0$, and so $\theta< 1$), then  the increase in $u$'s fractional degree was:
		\begin{align*}
		\sum_{v \in N_u(u)} \frac{(\theta - y_v)^{+}}{\beta} \left(1 + \frac{1-\theta}{f(\theta)}\right) =  \frac{f(\theta) + 1 - \theta}{\beta} = \frac{f(1-y_0) + y_0}{\beta} \geq \frac{y_0}{\beta}.
		\end{align*}
		
		\item  $y_0 = 0$ (thus $\theta = 1$), then the increase in $u$'s fractional degree was:
		\begin{align*}	
		\sum_{v \in N_u(u)} \frac{(\theta - y_v)^{+}}{\beta} \left(1 + \frac{1-\theta}{f(\theta)}\right) =  \sum_{v \in N_u(u)} \frac{(\theta - y_v)^{+}}{\beta}  \geq  0 = \frac{y_0}{\beta}.
		\end{align*}
		
	\end{enumerate}
	For every subsequent increase of the fractional degree due to a newly-arrived vertex we have that:
	\begin{align*}
	\frac{(\theta - y_u^{old})^+}{\beta} \left(1 + \frac{1-\theta}{f(\theta)}\right) & \geq \frac{(\theta - y_u^{old})^+}{\beta},
	\end{align*}
	Which concludes the proof for the lower bound.
	
	For the upper bound, by 
	\cite[Invariant 1]{wang2015two}, we have that 
	\begin{align}\label{eq:WW-bound-on-xu}
	\beta \cdot x_u & \leq y_{c} + f(1 - y_{0})  + \int_{y_{0}}^{y_{c}} \frac{1-x}{f(x)}\, dx.
	\end{align}
	This upper bound can be simplified by using \Cref{WW-tight}, as follows.
	Taking \eqref{eq:WW-bound-on-xu}, adding and subtracting $1 + f(1-y_u)$ and writing the integral $\int_{y_0}^{y_u} \frac{1-x}{f(x)}\, dx$ as the difference of two integrals  $\int_{y_0}^{1} \frac{1-x}{f(x)}\, dx$ - $\int_{y_u}^{1} \frac{1-x}{f(x)}\, dx$, and relying on \Cref{WW-tight}, we find that
	\begin{align*}
	\beta \cdot x_u & \leq y_{c} + f(1 - y_{0})  + \int_{y_{0}}^{y_{c}} \frac{1-x}{f(x)}\, dx \\
	&=  
	\left( 1 + f(1 - y_{0})  + \int_{y_{0}}^{1} \frac{1-x}{f(x)}\, dx \right)  -1  + y_{c} +
	\int_{1}^{y_{c}} \frac{1-x}{f(x)}\, dx \\
	&= \beta^*(f) + y_{c} - \left(  1 + f(1 - y_{c})  + \int_{y_{c}}^{1} \frac{1-x}{f(x)}\, dx  \right) + f(1-y_{c}) \\
	&= \beta^*{(f)} + y_{c} - \beta^*(f) + f(1-y_{c})\\
	& = y_{c} + f(1-y_{c}),
	\end{align*}
	from which the lemma follows.
\end{proof}

\section{Deferred Proofs of \Cref{sec:propertiesWW}}\label{app:WW-functions}

In this section we present the proofs deferred from \Cref{sec:propertiesWW}. We start by presenting a more manageable form for the function $f=f_{1+2\epsilon}$ which we use. 

A function in the WW family is determined by a parameter $k\geq 1$ and takes the following form
\begin{align*}
f_\kappa(\theta) = \left(\frac{1+\kappa}{2} - \theta\right)^{\frac{1+\kappa}{2\kappa}}\left(\theta+\frac{\kappa-1}{2}\right)^{\frac{\kappa-1}{2\kappa}}.
\end{align*}
Letting $\kappa= 1+2\varepsilon$, we get that $f:=f_\kappa$ is of the form
\begin{align*}
f(\theta) & = (1+\varepsilon - \theta)^{\frac{1+\varepsilon}{1+2\varepsilon}}\cdot \left(\theta + \varepsilon\right)^{\frac{\varepsilon}{1+2\varepsilon}} \\
&= \left(1+ \varepsilon - \theta\right) \cdot \left(
\frac{\theta+\varepsilon}{1+\varepsilon-\theta}
\right)^{\frac{\varepsilon}{1+2\varepsilon}}.
\end{align*}
Clearly this is water filling when $\varepsilon = 0$ and otherwise we have that the first term is like water filling and then the second term is less than $1$ for $z\leq 1/2$ and greater than $1$ if $z > 1/2$.

By Taylor expansion, we obtain the following more manageable form for $f$.
\begin{lem}
	\label{lm:taylor}
	There exists $\e_0\in (0,1)$ such that for every $\e\in (0,\e_0)$ and every $\theta\in [0,1]$, we have
	\begin{align*}
	f(\theta) \leq (1 - \theta) \left(1+\varepsilon
	\ln\left(\frac{\theta+\varepsilon}{1+\varepsilon-\theta} \right)\right) + 1.01\e.
	\end{align*}
\end{lem}
\begin{proof} 
	Taking the Taylor expansion of $e^x$, we find that
	\begin{equation*}
	\begin{split}
	f(\theta) & = \left(1+ \varepsilon - \theta\right) \cdot \left(
	\frac{\theta+\varepsilon}{1+\varepsilon-\theta}
	\right)^{\frac{\varepsilon}{1+2\varepsilon}}\\
	& = 
	\left(1+ \varepsilon - \theta \right) \cdot \sum_{i=0}^\infty 
	\frac{
		\left( \ln \left(\frac{\theta + \varepsilon}{1+ \varepsilon - \theta}\right) 
		\cdot \frac{\varepsilon}{1+2\varepsilon}
		\right)^i
	}{i!} \\
	& = (1+\varepsilon - \theta) \left(1+ \ln \left(\frac{\theta+\varepsilon}{1+ \varepsilon - \theta} \right) \cdot\frac{\varepsilon}{1+2\varepsilon}\right)  +  o(\varepsilon) \\
	& = (1+\varepsilon - \theta ) + (1-\theta) \ln \left(\frac{\theta+\varepsilon}{1+ \varepsilon - \theta} \right) \cdot\frac{\varepsilon}{1+2\varepsilon} + o(\varepsilon) \\
	& = (1+\varepsilon - \theta ) + (1- \theta) \varepsilon
	\ln\left(\frac{\theta+\varepsilon}{1+\varepsilon-\theta} \right) + o(\varepsilon)\\
	& = (1 - \theta )\left(1 + \varepsilon
	\ln\left(\frac{\theta+\varepsilon}{1+\varepsilon-\theta} \right)\right) +\e+ o(\varepsilon).	
	\end{split}
	\end{equation*}
	To be precise, for $\theta\in [0,1]$ and $0 < \epsilon\leq \epsilon_0\leq 1$ (implying for example $\frac{\theta+\epsilon}{1+\epsilon-\theta} \leq \frac{2}{\epsilon}$), we will show that terms dropped in the third, fourth and fifth lines are all
	at most some $O((\ln(\frac{1}{\epsilon})\cdot \epsilon)^2) = o(\epsilon)$, from which the lemma follows as the sum of these terms is at most $0.01\epsilon$ for $\e\leq \e_0$ and $\epsilon_0$ sufficiently small.
	
	Indeed, in the third line, we dropped 
	\begin{align*}
	\left(1+ \varepsilon - \theta\right) \cdot \sum_{i=2}^\infty 
	\frac{
		\left( \ln \left(\frac{\theta+\varepsilon}{1+ \varepsilon - \theta}\right) 
		\cdot \frac{\varepsilon}{1+2\varepsilon}
		\right)^i
	}{i!} \leq 2\cdot \sum_{i=2}^\infty \frac{(\ln(\frac{2}{\epsilon})\cdot \eps)^i}{i!} \leq \cdot \sum_{i=2}^\infty \frac{\left(\ln\left(\frac{2}{\epsilon}\right)\cdot \eps\right)^i}{i^2} = O((\ln(1/\e)\cdot \eps)^2),
	\end{align*} 
	where the last step used that $\ln(1/\eps)\cdot \epsilon\leq 1$ holds for all $\epsilon\geq 0$.
	In the fourth line, we dropped 
	\begin{align*}\epsilon\cdot \ln\left(\frac{\theta+\epsilon}{1+\epsilon-\theta}\right)\cdot \frac{\epsilon}{1+2\epsilon} \leq \epsilon^2\cdot \ln\left(2/\e\right) = O((\ln(1/\e)\cdot \eps)^2).
	\end{align*}
	Finally, in the fifth line, we dropped 
	\begin{align*}
	(1-z)\cdot \left(\epsilon-\frac{\epsilon}{1+2\epsilon}\right)\cdot \ln\left(\frac{\theta+\epsilon}{1+\epsilon-\theta}\right) & \leq 1\cdot (\epsilon^2/(1+2\eps))\cdot \ln\left(2/\e\right) = O((\ln(1/\e)\cdot \eps)^2).\qedhere
	\end{align*} 
\end{proof}

Given this more manageable form for $f$, we can now turn to prove \Cref{lem:propertiesWW}, restated below.

\WWproperties*

\begin{proof}
	We begin by getting a generic upper bound for $z_u$. We note that each edge $e$ is matched by \Cref{algo:wwrounding} with probability at most $x_e$ by \Cref{line:second-pick-drop}. Therefore, $u$ is matched before $v$ arrives with probability at most $x_u := \sum_{w\in N_v(u)\setminus \{v\}} x_{wu}$, the fractional degree of $u$ before $v$ arrives.
	Therefore, by \Cref{refined-ww-xu-bound}, the probability that $u$ is free is at least
\begin{equation}
\Pr[u \mbox{ free when $v$ arrives}] \geq 1-x_u \geq 1-\frac{y_u + f(1-y_u)}{\beta},
\end{equation}
from which, together with the definition of $x_{uv} = \frac{1}{\beta}(\theta - y_u)^+ \left(1 + \frac{1-\theta}{f(\theta)}\right)$, we obtain the following upper bound on $z_u$:
\begin{equation}\label{eq:z-upper-bound}
z_u = \frac{x_{uv}}{\Pr[\text{$u$ is free when $v$ arrive}]} \leq  \frac{\frac{1}{\beta}(\theta - y_u)^+ \left(1 + \frac{1-\theta}{f(\theta)}\right)}{1 - \frac{y_u+f(1-y_u)}{\beta}} = \frac{(\theta - y_u)\left(1 + \frac{1-\theta}{f(\theta)}\right)}{\beta - (y_u + f(1-y_u))}.
\end{equation}
		
We start by upper bounding $\sum_{u\in N_v(v)} z_u$, giving a bound which will prove useful in the proofs of both {\bf (1)} and {\bf(2)}. Recall that $\theta$ is defined as the largest $\theta\leq 1$ such that 
\begin{equation}\label{eq:def-theta}
	\sum_{u\in N_v(v)} (\theta-y_u)^+ \leq f(\theta).
\end{equation}

Summing \eqref{eq:z-upper-bound} over all $u\in N_v(v)$, we find that
	\begin{align*}
	\sum_{u\in N_v(v)} z_u& 
	\leq \sum_{u\in N_v(v)}\frac{(\theta-y_u)^+\cdot (1+\frac{1-\theta}{f(\theta)})}{\beta-(\theta+f(1-\theta))}&& \text{($f(\cdot)$ is non-increasing, by~\Cref{f-non-increasing})}\\
	&\leq \frac{f(\theta)+1-\theta}{\beta-(\theta+f(1-\theta))} && \text{(by \eqref{eq:def-theta} and $\beta\geq \beta^*(f) = 1+f(0) \geq \theta+f(1-\theta)$)}
	\end{align*}
We therefore wish to upper bound $\frac{f(\theta) + 1-\theta}{\beta - \theta - f(1-\theta)}$. To this end let $\gamma(\theta, \e) :=\e \ln \left(\frac{\theta+\e}{1+\e-\theta}\right)$. Before proceeding to the proof, it would be useful to summarize some properties of the function $\gamma(\theta, \e)$. 
\begin{enumerate}
\item $ \gamma(\theta, \e) = - \gamma(1-\theta, \e) $ for all $\theta \in [0,1]$ \label{gamma:property4}.
\item For $c$, $\e_0$ sufficiently small we have
for all $\theta \in [0,c)$ that $\gamma(\theta,\eps) \leq \e \ln \left(\frac{c+\e}{1+\e-c}\right) \leq -20\cdot \eps$, and
 for all $\theta \in (1-c,1]$ that $ \gamma(\theta, \e) \geq \e \ln \left(\frac{1-c+\e}{1+\e-(1-c)}\right) \geq 20\cdot \eps$. \label{gamma:property1}
\item $\gamma (\theta, \e)\cdot (1 - 2 \theta) \leq 0$ for $\theta \in [0,1]$, since $\gamma(\theta,\epsilon)\leq 0$ for $\theta\leq 1/2$ and $\gamma(\theta,\epsilon)\geq 0$ for $\theta \geq 1/2$.
 \label{gamma:property2}
\item $\theta \cdot  \gamma (\theta, \e)  \geq -\e$ for all $\theta \in [0,1]$. \label{gamma:property3}

The last property follows from $ \ln \left(\frac{1+\e-\theta}{\theta+\e}  \right) \leq \ln \left(\frac{1+\e+\theta}{\theta+\e}  \right) \leq \ln \left(1 + \frac{1}{\theta+\e}  \right) \leq \frac{1}{\theta + \e}  \leq \frac{1}{\theta}$, which implies in particular that $\theta \cdot \gamma(\theta, \e) =  \theta \cdot \e \cdot \left(- \ln \left( \frac{1+\e-\theta}{\theta + \e}\right) \right)  \geq -\e$.
%\end{proof}
\end{enumerate}

We will use $\gamma$ as shorthand for  $\gamma(\theta, \e)$.  Recalling that $\beta=2-\e$ and using~\Cref{lm:taylor},  we have:
	\begin{equation}\label{eq:2isdighg}
	\begin{split}
	\frac{f(\theta)+1-\theta}{\beta-(\theta+f(1-\theta))}&\leq \frac{(1-\theta)\left(1+\e \ln \left(\frac{\theta+\e}{1+\e-\theta}\right)\right)-\theta+1+1.01\e}{2-\e - \theta-\theta\left(1+\e \ln \left(\frac{1-\theta+\e}{\theta+\e}\right)\right)-1.01\e}\\	
	& \leq \frac{(1-\theta)(2+\gamma)+2\e}{2-2\theta+\theta \gamma-3\e}\\
	& =	1+\frac{\gamma(1-2\theta)+5\e}{2-2\theta+\theta\gamma-3\e}.
\end{split}
\end{equation} 

We will continue by proving that the second term is negative. First we prove that the denominator is positive. To this end, first consider the case when $\theta \in [0,c)$. In this case for $\e_0, c$ sufficiently small one has that:  $2-2\theta+\theta\gamma-2\e > 2 - 2 \theta -\e -2 \e > 0$ from \cref{gamma:property3}. Moreover, when  $\theta \in (1-c, 1]$ one has that $\theta > \frac{1}{2}$ (since $c$ is small) and $\gamma \geq 20 \e$ from \cref{gamma:property1}. Thus $2-2\theta+\theta\gamma-2\e \geq \theta \gamma - 2 \e \geq \frac{1}{2} \cdot 20 \e -3 \e = 7 \e > 0 $. Now, it remains to prove that the numerator is always negative. When $\theta \in [0,c)$ we have that $1 - 2 \theta \geq 3 /4$(since $c$ is small) and $\gamma \leq -20 \e$ from \cref{gamma:property1}, therefore $\gamma (1- 2 \theta) + 5\e \leq \gamma \cdot \frac{3}{4}+5 \cdot (-\frac{\gamma}{20}) = \frac{\gamma}{2} < 0 $. In the case where $\theta \in (1-c, 1]$, we have that $1 - 2 \theta <- 3/4$, and $\theta > 1/2 $ (since $c$ is small), and $\gamma \geq 20 \e$ from \cref{gamma:property1}, thus $\gamma (1 - 2 \theta) + 5 \e \leq -\frac{3}{4} \cdot 20 \e  + 5\e = -10 \e < 0 $.

	We now turn to {\bf (2)}.  We assume that $\theta\in (c, 1-c)$, since otherwise the claim is trivial, by {\bf (1)}. We have by~\eqref{eq:2isdighg} that  $\frac{f(\theta)+1-\theta}{\beta-(\theta-f(1-\theta))} \leq 	1+\frac{\gamma(1-2\theta)+5\e}{2-2\theta+\theta\gamma-3\e}$. We have that  $\gamma (1 - 2 \theta ) + 5 \e \leq 5 \e$ from \cref{gamma:property2}. Furthermore, using \cref{gamma:property3} we have that $2 - 2 \theta + \theta \gamma - 3 \e  \geq 2c + -4 \e > c$ for a sufficiently small $\e_0$. Overall, the second term is bounded above by $\frac{5}{c} \cdot \e < C \cdot \e$, for $C > \frac{5}{c} > \frac{1}{c}$ as required.

	We now prove {\bf (3)}. Note that by \textbf{(1)}, $\sum_{u\in N_v(v)}z_u > 1$ implies that $\theta\in (c, 1-c)$. Now, for every $u\in N_v(v)$, let $\alpha_u:=\frac{(\theta-y_u)^+}{f(\theta)}$, so that $y_u=\theta-f(\theta)\cdot \alpha_u$ if $y_u \leq \theta$.
	We also note that by definition of $\alpha_u$ and our choice of $\theta$, we have $\sum_{u\in N_v(v)} \alpha_u = \sum_{u\in N_v(v)} \frac{(\theta-y_u)^+}{f(\theta)} \leq 1$. 
	In the proof of {\bf(3)} and {\bf(4)} we will assume for notational simplicity that all $u\in N_v(v)$ have $y_u \leq \theta$, implying $z_u\geq 0$. Summing up \eqref{eq:z-upper-bound} over all $u\in N_v(v)$ and substituting in $\alpha_u$, we thus find that
	\begin{align*}
	\sum_{u\in N_v(v)} z_u&\leq \sum_{u\in N_v(v)} \frac{(\theta-y_u)^+(1+\frac{1-\theta}{f(\theta)})}{\beta-(y_u+f(1-y_u))}\\
	&=\sum_{u\in N_v(v)} \alpha_u \cdot  \frac{f(\theta)+1-\theta}{\beta-(y_u+f(1-y_u))}\\
	&\leq \sum_{u\in N_v(v)} \alpha_u \cdot  \frac{f(\theta)+1-\theta}{2-y_u-f(1-y_u)-2.01\e} && \text{(by~\Cref{lm:taylor} and $\beta=2-\eps$)}\\
	& \leq \sum_{u\in N_v(v)} \alpha_u \cdot \frac{f(\theta)+1-\theta}{2- 4\e -2y_u},
	\end{align*}
	In the last transition we used again (as in \cref{gamma:property3}) that $ y_u\cdot \e \ln \left(\frac{1-y_u+\e}{y_u+\e}\right) \leq \e$, which implies $f(\theta) \leq 1-\theta + \e$ for all $\theta\in [0,1]$. Substituting $y_u = \theta-f(\theta)\cdot \alpha_u$ into the above upper bound on $\sum_{u\in N_v(v)} z_u$, we get
	\begin{align}
	\sum_{u\in N_v(v)} z_u&\leq \sum_{u\in N_v(v)} \alpha_u \cdot \frac{f(\theta)+1-\theta}{2-4\eps-2\theta+2 f(\theta) \cdot \alpha_u} \nonumber \\
	&=\sum_{u\in N_v(v)} \alpha_u \cdot \frac{f(\theta)+1-\theta}{2-4\e-2\theta} - \sum_{u\in N_v(v)} \frac{(f(\theta)+1-\theta)\cdot 2f(\theta)\cdot \alpha_u^2}{(2-4\e-2\theta)\cdot (2-4\e-2\theta+2f(\theta)\cdot \alpha_u)}, \label{eq:23gewgewg}
	\end{align}
	using the elementary identity $\frac{1}{a+b} = \frac{1}{a} - \frac{b}{a(a+b)}$ for appropriate $a$ and $b$.
	Now, both terms in the last line of \eqref{eq:23gewgewg} can be significantly simplified, as follows. For the former term, again using that $f(\theta)\leq 1-\theta+\epsilon$, together with $\sum_{u\in N_v(v)} \alpha_u \leq 1$ noted above, we find that 
	\begin{align} \label{first-term-23gewgewg}
	\sum_{u\in N_v(v)}\alpha_u\cdot \frac{f(\theta) + 1-\theta}{2-4\epsilon-2\theta} & \leq \sum_{u\in N_v(v)}\alpha_u\cdot  \frac{{2+\epsilon-2\theta}}{{2-4\epsilon-2\theta}} = \sum_{u\in N_v(v)}\alpha_u\cdot \left( 1+\frac{5\eps}{2-4\epsilon-2\theta}\right) \leq 1+O(\eps),
	\end{align}
	where in the last step we used that $\theta\leq 1-c$ and $c$ is some fixed constant.
	For the second term in the last line of \eqref{eq:23gewgewg}, we note that 
	\begin{align}\label{second-term-23gewgewg}
		\sum_{u\in N_v(v)} \frac{(f(\theta)+1-\theta)\cdot 2f(\theta)\cdot \alpha_u^2}{(2-4\e-2\theta)\cdot (2-4\e-2\theta+2f(\theta)\cdot \alpha_u)} & = \Omega(1)\cdot \left(\sum_{u\in N_v(v)} \alpha_u^2 \right).
	\end{align}
	To see this, first note that for $\theta\in (c,1-c)$, the numerator of each summand of the LHS is at least $2f(c)^2 \cdot \alpha_u^2 \geq \Omega(\alpha_u^2)$, since $f$ is decreasing by \Cref{f-non-increasing} and $f(c)\geq \frac{1}{2}\cdot (1+\eps-c) \geq \Omega(1)$ for $c$ and $\eps$ sufficiently small.
	To verify the first inequality of this lower bound for $f(c)$, recall that 
	$f(c) = \left(1+ \varepsilon - c\right) \cdot \left(
	\frac{c+\varepsilon}{1+\varepsilon-c}
	\right)^{\frac{\varepsilon}{1+2\varepsilon}}$. Now, for $\epsilon$ tending to zero and $c<1/2$, the term $\left(
	\frac{\theta+\varepsilon}{1+\varepsilon-\theta}
	\right)^{\frac{\varepsilon}{1+2\varepsilon}}$ tends to one as $\epsilon$ tends to zero. Therefore for $\epsilon$ sufficiently small we have $f(c)\geq \frac{1}{2}\cdot (1+\epsilon-c)$ for all $c < 1/2$.
	We now turn to upper bounding the denominator of 
	each summand in the LHS of \Cref{second-term-23gewgewg}.
	Indeed, substituting $y_u = \theta-f(\theta)\cdot \alpha_u$, we find that each such denominator is at most
	$(2-4\e-2\theta)\cdot (2-4\e-2\theta+2f(\theta)\cdot \alpha_u) \leq (1/2)\cdot (2-4\eps - 2y_u) \leq (1/2)\cdot(2-4\eps - 2c) \leq O(1)$ for $c$ and $\epsilon$ sufficiently small. Note that both numerator and denominator are positive for sufficiently small $c$ and $\epsilon_0$.
	Substituting the bounds of \eqref{first-term-23gewgewg} and \eqref{second-term-23gewgewg} into \eqref{eq:23gewgewg}, we obtain
	\begin{align}\label{z-u-alpha-u-bound}
	\sum_{u\in N_v(v)} z_u
	& \leq 1 + O(\e) -\Omega(1)\cdot \left(\sum_{u\in N_v(v)} \alpha_u^2\right).
	\end{align}

From \cref{z-u-alpha-u-bound} and $\sum_{u\in N_v(v)} z_u > 1$ by assumption of {\bf(3)}, we get that

\begin{equation}\label{eq:032utg23g}
\sum_{u\in N_v(v)} \alpha_u^2\leq C^{'} \e
\end{equation} 
for an absolute constant $C^{'}>1$, since otherwise $ \sum_{u \in N_v(v)} z_u\leq 1$. Finally, it remains to note that 
	\begin{align*}
	\sum_{u\in N_v(v)} z_u^2&=\sum_{u\in N_v(v)} \left(\frac{\alpha_u  \cdot (f(\theta)+1-\theta)}{\beta-(y_u+f(1-y_u))}\right)^2\\
	&\leq \left(\sum_{u\in N_v(v)} \alpha_u^2\right)\cdot \left(\frac{f(\theta)+1-\theta}{\beta-(\theta+f(1-\theta))}\right)^2 && \text{(by \Cref{f-non-increasing} and $y_u \leq \theta$)}\\
	&\leq \left(\sum_{u\in N_v(v)} \alpha_u^2\right)\cdot \left(\frac{f(\theta)+1-\theta}{\beta-(1-c+f(c))}\right)^2 && \text{(by \Cref{f-non-increasing} and $\theta\leq 1-c$)}\\	
	&\leq \left(\sum_{u\in N_v(v)} \alpha_u^2\right)\cdot \left(\frac{1-\theta+\eps+1-\theta}{\beta-(1-c+1-c+\eps)}\right)^2 && \text{($f(c)\leq 1-c+\eps$)}\\		
	& \leq \left(\sum_{u\in N_v(v)} \alpha_u^2\right)\cdot \frac{2}{2c-2\e}\\
	& \leq C \e,
	\end{align*}
	for some constant $C\geq \frac{2}{2c-2\epsilon}$. Thus $z_u^2 \leq \sum_{u\in N_v(v)} z_u \leq C\e$ and so $z_u \leq \sqrt{C \cdot \e } \leq C \sqrt{\e}$, as claimed. 
	
	We now prove {\bf (4)}. Since $\sum_{u \in N_v(v)} z_u > 1$ implies $\theta\in (c, 1-c)$ by {\bf (1)}, using the definition of $\alpha_u$'s from the proof of {\bf (3)} together with the fact that $\alpha_u\leq C^{'} \sqrt{\e}$ for every $u\in N_v(v)$ by~\eqref{eq:032utg23g} and the fact that $f(\theta)\leq 2$ for all $\theta\in [0,1]$ (by \Cref{lm:taylor}), we get that
	\begin{equation*}
	\begin{split}
	y_u=\theta-f(\theta)\cdot \alpha_u\in [c-O(\sqrt{\e}), 1-c]\subseteq [c/2, 1-c/2],
	\end{split}
	\end{equation*}
	for sufficiently small $\e_0>0$, as required.

    As for {\bf (5)}, simplifying \eqref{eq:z-upper-bound} and using the fact that $\theta - y_u \leq f(\theta)$, we get
\begin{align*}
  z_u &\leq \frac{\theta - y_u + 1 - \theta}{\beta - y_u - f(1 - y_u)} = \frac{1 - y_u}{\beta - y_u - f(1 - y_u)}.
\end{align*}

Recall from~\Cref{lm:taylor} that for all $\theta\in [0,1]$, we have $f(\theta) \leq (1 - \theta) \left(1+\varepsilon
\ln\left(\frac{\theta+\varepsilon}{1+\varepsilon-\theta} \right)\right) + 1.01\e$, 
which implies the following:
	\begin{enumerate}
	    \item \label{point1} For all $\theta \in [0, 1]$, we have $f(\theta) \leq 1 - \theta + \sqrt{\e}$, and
	    \item \label{point2} For $\theta < e^{-10}$, we have $f(\theta) \leq (1-\theta)(1 + \e(\ln( (e^{-10} + \e)/(1 - e^{-10} + \e)) + 1.01 \e \leq 1 - \theta - 2\e$.
	\end{enumerate}
Suppose that $y_u \leq 1 - e^{-10}$. Then using \Cref{point1}, we have
\begin{align*}
  z_u &\leq \frac{1-y_u}{\beta - y_u - f(1-y_u)} \leq \frac{1 - y_u}{2 - \e - y_u - y_u - \sqrt{\e}} \\
  &\leq \frac{ 1 - y_u} {2(1 - y_u) - 2 \sqrt{\e}} \leq 1/2+\frac{2\sqrt{\epsilon}}{2e^{-10}-2\sqrt{\epsilon}} \leq 1/2 + O(\sqrt{\e}).
\end{align*}
Now suppose that $y_u > 1 - e^{-10}$. Then $1 - y_u < e^{-10}$, and so by~\Cref{point2}, $f(1 - y_u) \leq 1 - y_u - 2\e$.
Thus we have
\begin{align*}
  z_u &\leq \frac{1-y}{\beta - y_u - f(1 -y_u)} \leq \frac{1 -y_u}{2 - \e - y_u - (y_u - 2\e)} = \frac{1-y_u}{2(1 - y_u) + \e} \leq 1/2,
\end{align*}
completing the proof.
\end{proof}

Finally, we rely on \Cref{lm:taylor} to prove that the fractional solution maintained by \Cref{line:fractional-to-round} is $1/\beta$ competitive, as implied by \Cref{ww-approx} and the following restated fact.

\betafine*
\begin{proof}
	Let us denote as before $f=f_{1+2\epsilon}$.
	Recall that $\beta^*(f) = 1+f(0)$. By \Cref{lm:taylor}, this is at most $1+f(0) \leq 1 + \left(1+\varepsilon
	\ln\left(\frac{\varepsilon}{1+\varepsilon} \right)\right) + 1.01\e$. 
	But for small enough $\epsilon$, we have that $\ln\left(\frac{\varepsilon}{1+\varepsilon} \right) \leq  -2.01$, implying that $1+f(0) \leq 2-\epsilon$, as claimed.
\end{proof}

\bibliographystyle{acmsmall}
\bibliography{refs.bib}

\end{document}